\documentclass[conference,letterpaper]{IEEEtran}
\let\proof\relax   

\usepackage[T1]{fontenc}
\usepackage{mathpple}
\usepackage{times}
\usepackage{amsthm,xpatch}
\usepackage{amsmath,amsfonts}
\usepackage{cite}
\usepackage{amssymb}
\usepackage{dsfont}
\usepackage{graphicx, subfigure}
\usepackage{color}
\usepackage{breqn}
\usepackage{mathtools}
\usepackage{bbm}
\usepackage{latexsym}
\usepackage[ruled, linesnumbered]{algorithm2e}
\usepackage{accents}
\usepackage{rotating}
\usepackage{tikz}
\usepackage[mathscr]{euscript}
\usepackage{bm}
\usepackage[frak=lucida]{mathalpha} 
\usepackage{comment}
\usepackage{balance}
\usepackage{float}
\usepackage{stfloats}
\renewcommand{\theequation}{\upshape\arabic{equation}}

\makeatletter
\def\maketag@@@#1{\hbox{\m@th\normalfont#1}}
\makeatother

\newtheorem{lemma}{Lemma}
\newtheorem{theorem}{Theorem}

\newtheorem{remark}{Remark}

\newtheorem*{example*}{Example}

\makeatletter
\newcommand*{\transpose}{%
  {\mathpalette\@transpose{}}%
}
\newcommand*{\@transpose}[2]{%
  \raisebox{\depth}{$\m@th#1\intercal$}%
}
\makeatother

\IEEEoverridecommandlockouts

\begin{document}

\makeatletter
\newcommand{\raisemath}[1]{\mathpalette{\raisem@th{#1}}}
\newcommand{\raisem@th}[3]{\raisebox{#1}{$#2#3$}}
\makeatother

\newcommand{\mstk}{\hspace{-0.145cm}*}

\newcommand{\mstl}{\hspace{-0.105cm}*}

\newcommand{\mstm}{\hspace{-0.175cm}*}

\newcommand{\SB}[3]{
\sum_{#2 \in #1}\biggl|\overline{X}_{#2}\biggr| #3
\biggl|\bigcap_{#2 \notin #1}\overline{X}_{#2}\biggr|
}

\newcommand{\Mod}[1]{\ (\textup{mod}\ #1)}

\newcommand{\overbar}[1]{\mkern 0mu\overline{\mkern-0mu#1\mkern-8.5mu}\mkern 6mu}

\makeatletter
\newcommand*\nss[3]{%
  \begingroup
  \setbox0\hbox{$\m@th\scriptstyle\cramped{#2}$}%
  \setbox2\hbox{$\m@th\scriptstyle#3$}%
  \dimen@=\fontdimen8\textfont3
  \multiply\dimen@ by 4             % 4x the default rule thickness
  \advance \dimen@ by \ht0
  \advance \dimen@ by -\fontdimen17\textfont2
  \@tempdima=\fontdimen5\textfont2  % x-height
  \multiply\@tempdima by 4
  \divide  \@tempdima by 5          % 80% of the x-height
  % Modifications are only necessary if the top of the subscript is not that high:
  \ifdim\dimen@<\@tempdima
    \ht0=0pt                        % don't let the subscript interfere
    \@tempdima=\fontdimen5\textfont2
    \divide\@tempdima by 4          % 25% of the x-height
    \advance \dimen@ by -\@tempdima % if >0, add to depth of superscript!
    \ifdim\dimen@>0pt
      \@tempdima=\dp2
      \advance\@tempdima by \dimen@
      \dp2=\@tempdima
    \fi
  \fi
  #1_{\box0}^{\box2}%
  \endgroup
  }
\makeatother

\makeatletter
\renewenvironment{proof}[1][\proofname]{\par
  \pushQED{\qed}%
  \normalfont \topsep6\p@\@plus6\p@\relax
  \trivlist
  \item[\hskip\labelsep
        \itshape
%    #1\@addpunct{.}]\ignorespaces% DELETED
    #1\@addpunct{:}]\ignorespaces% ADDED
}{%
  \popQED\endtrivlist\@endpefalse
}
\makeatother

\makeatletter
\newsavebox\myboxA
\newsavebox\myboxB
\newlength\mylenA

\newcommand*\xoverline[2][0.75]{%
    \sbox{\myboxA}{$\m@th#2$}%
    \setbox\myboxB\null% Phantom box
    \ht\myboxB=\ht\myboxA%
    \dp\myboxB=\dp\myboxA%
    \wd\myboxB=#1\wd\myboxA% Scale phantom
    \sbox\myboxB{$\m@th\overline{\copy\myboxB}$}%  Overlined phantom
    \setlength\mylenA{\the\wd\myboxA}%   calc width diff
    \addtolength\mylenA{-\the\wd\myboxB}%
    \ifdim\wd\myboxB<\wd\myboxA%
       \rlap{\hskip 0.5\mylenA\usebox\myboxB}{\usebox\myboxA}%
    \else
        \hskip -0.5\mylenA\rlap{\usebox\myboxA}{\hskip 0.5\mylenA\usebox\myboxB}%
    \fi}
\makeatother

\xpatchcmd{\proof}{\hskip\labelsep}{\hskip3.75\labelsep}{}{}

\pagestyle{empty}

\title{\fontsize{22.5}{28}\selectfont  Achieving Capacity of PIR with Private Side Information\\ with Low Sub-packetization and without MDS Codes}

\author{Leila Erhili and Anoosheh Heidarzadeh\thanks{Leila Erhili is with the Department of Computer Science and Engineering, Santa Clara University, Santa Clara, CA 95053 USA (E-mail: lerhili@scu.edu). Anoosheh Heidarzadeh is with the Department of Electrical and Computer Engineering, Santa Clara University, Santa Clara, CA 95053 USA (E-mail: aheidarzadeh@scu.edu).}
}

\maketitle 

\thispagestyle{empty}

\begin{abstract}
This paper revisits the problem of multi-server Private Information Retrieval with Private Side Information (PIR-PSI). 
In this problem, $N$ non-colluding servers store identical copies of $K$ messages, each comprising $L$ symbols from $\mathbbmss{F}_q$, 
and 
a user, who knows $M$ of these messages, wants to retrieve one of the remaining $K-M$ messages. 
The user's goal is to retrieve the desired message by downloading the minimum amount of information from the servers while revealing no information about the identities of the desired message and side information messages to any server. 
The capacity of PIR-PSI, defined as the maximum achievable download rate, was previously characterized for all $N$, $K$, and $M$ when $L$ and $q$ are sufficiently large---specifically, growing exponentially with $K$, to ensure the divisibility of each message into $N^K$ sub-packets and to guarantee the existence of an MDS code with its length and dimension being exponential in $K$. 
In this work, we propose a new capacity-achieving PIR-PSI scheme that is applicable to all $N$, $K$, $M$, $L$, and $q$ where $N\geq M+1$ and $N-1\mid L$. 
The proposed scheme operates with a sub-packetization level of $N-1$, independent of $K$, and works over any finite field without requiring an MDS code. 
\end{abstract}

\section{Introduction}
In the classical setting of the Private Information Retrieval (PIR) problem, 
a user wishes to retrieve one message from a dataset of messages with copies stored on a single or multiple remote servers. 
The objective is to design a retrieval mechanism that minimizes the amount of downloaded information while revealing no information about the identity of the desired message. 
While achieving privacy in the single-server setting requires downloading the entire dataset~\cite{CGKS1995}, 
the multi-server setting offers more efficient solutions in terms of download rate~\cite{SJ2017,SJ2016ArbitraryTIFS,SJ2018Multiround,TSC2019,VBU2022}. 
This has ignited extensive studies on various multi-server PIR settings in recent years, 
including 
PIR with colluding servers~\cite{SJ2018Colluding,BU2019Colluding,ZX2019,YLK2020,LJJ2021,HFLH2022,ZTSP2022}, 
PIR with coded databases~\cite{TGKHHER2017,BU18,BAWU2020}, and multi-message PIR (MPIR)~\cite{BU2018,WHS2022}. 

In parallel to research on multi-server settings, a separate line of work has emerged to improve the efficiency of the retrieval process in the single-server setting. 
Within this line of research, 
the main idea has been to provide the user with a prior side information about the dataset. 
Several variations of single-server PIR with side information have been studied in the literature, 
see, e.g.,~\cite{KGHERS2017No0,HKGRS2018,LG2018,HKS2018,HKS2019,HKRS2019,HKS2019Journal,KKHS32019,KHSO2021,HS2021,HS2022Reuse,LJ2022,GLH2022,HS2022LinCap}. 
The application of side information has also been extended to multi-server settings. 
This includes 
cache-aided PIR~\cite{T2017,WBU2018,WBU2018No2,ZWSJC2021}, 
multi-server PIR with side information (PIR-SI)~\cite{KKHS12019,KGHERS2020,LG2020CISS,WHS2024}, 
multi-server PIR with private side information (PIR-PSI)~\cite{KKHS22019,CWJ2020,KH2021Journal}, and
multi-server MPIR with private side information (MPIR-PSI)~\cite{SSM2018}.

In this work, we revisit the multi-server PIR-PSI problem---simply referred to as PIR-PSI---which was initially studied in~\cite{CWJ2020} and later explored in~\cite{KH2021Journal}. 
In this problem, there are $N$ non-colluding servers, storing identical copies of $K$ messages, each of which consists of $L$ symbols from $\mathbbmss{F}_q$, and there is a user who wants to retrieve one of these messages. 
The user initially knows a subset of $M$ out of $K$ messages---distinct from the desired message---as a prior side information, and the servers are initially unaware of the identities of these messages. 
The user's goal is to retrieve the desired message by downloading the minimum amount of information from the servers 
while preserving the privacy of both the desired message's identity and the identities of the side information messages against the servers. 

Chen \emph{et al.} in~\cite{CWJ2020} established a converse bound on the capacity of PIR-PSI---defined as the maximum achievable download rate over all PIR-PSI schemes---for all values of $N$, $K$, $M$, $L$, and $q$. 
They also showed that the converse bound is tight when $L$ and $q$ are sufficiently large. 
In particular, the achievability scheme presented in~\cite{CWJ2020} relies on the following two assumptions: 
the message length $L$ is an integer multiple of $N^K$, guaranteeing that each message can be divided into $N^K$ sub-packets of equal size; and 
the field size $q$ is sufficiently large---scaling with $N^K$---ensuring the existence of an MDS code of length $(2N^K-N^M-1)/(N-1)$ and dimension $(N^K-1)/(N-1)$. 
This implies that the sub-packetization level and the field size must grow exponentially as the number of messages increases. 
Such exponential growth, however, may impede the application of this scheme in practice. 
Motivated by this, 
Krishnan and Harshan introduced a scheme in~\cite{KH2021Journal} that operates over the binary field, eliminating the need for an MDS code. 
However, the scheme in~\cite{KH2021Journal} relies on a sub-packetization level growing exponentially with $K$, is limited  to the case of $N=2$ and $M=2$, and is not capacity-achieving.   

In this work, we propose a new capacity-achieving PIR-PSI scheme that is applicable to all values of $N$, $K$, $M$, $L$, and $q$ where $N\geq M+1$ and $N-1\mid L$.
Our scheme works with a constant sub-packetization level regardless of $K$ and operates over any finite field without requiring an MDS code. 
In this scheme, the user downloads a single coded sub-packet---the sum of a specific subset of message sub-packets---from each server and  
recovers each sub-packet of the desired message from one distinct server. 
To determine the number and selection of message sub-packets contributing to each coded sub-packet, we employ a randomized algorithm. 
This algorithm utilizes a carefully designed non-uniform probability distribution to guarantee an equal probability for all coded sub-packets with the same number of message sub-packets.

\section{Problem Setup}
We represent random variables and their realizations by bold-face symbols and regular symbols, respectively. 
For an integer $i\geq 1$, 
the set $\{1,\dots,i\}$ is denoted by $[i]$, and for integers $i,j\geq 1$ such that $i<j$, 
the set $\{i,\dots,j\}$ is denoted by $[i:j]$. 
We denote the positive part of $i$ by $i^{+}$, i.e., $i^{+} = i$ if $i\geq 0$, and $i^{+}=0$ otherwise. 
We denote the $i$th component and the support of a vector $\mathrm{v}$ by $\mathrm{v}(i)$ and $\mathrm{supp}(\mathrm{v})$, respectively. 
For ease of notation, 
we define $\binom{n}{k} := 0$ for $k<0$ or $k>n$. 

Consider $N>1$ servers, each storing an identical copy of $K>1$ messages ${\mathrm{X}_1,\dots,\mathrm{X}_K}\in \mathbbmss{F}^L_q$, 
where the random variables $\mathbf{X}_1,\dots,\mathbf{X}_K$ are independent and uniformly distributed over $\mathbbmss{F}^{L}_{q}$. 
Here, $\mathbbmss{F}^{L}_q$ is the $L$-dimensional vector space for $L\geq 1$ over a finite field $\mathbbmss{F}_q$ for a prime power $q\geq 2$. 
That is, each message $\mathrm{X}_i$ consists of $L$ symbols from $\mathbbmss{F}_p$. 

Consider a user who wants to retrieve the message $\mathrm{X}_{\mathrm{W}}$ for some ${\mathrm{W}\in [K]}$, 
where the random variable $\mathbf{W}$ is uniformly distributed over ${[K]}$.
Suppose that the user knows $M\geq 1$ messages ${\mathrm{X}_{\mathrm{S}}:=\{\mathrm{X}_i\}_{i\in \mathrm{S}}}$ for some $M$-subset ${\mathrm{S}\subset [K]\setminus \{\mathrm{W}\}}$, 
where the random variable $\mathbf{S}$ is uniformly distributed over all $M$-subsets of $[K]\setminus \{\mathrm{W}\}$.
We refer to $\mathrm{X}_{\mathrm{W}}$ as the \emph{user's demand}, 
$\mathrm{X}_{\mathrm{S}}$ as the \emph{user's side information}, and 
$\{\mathrm{X}_i\}_{i\in [K]\setminus \{\mathrm{W}\}\cup \mathrm{S}}$ as the \emph{interference messages}.  

For each ${n\in [N]}$, the user generates a query $\mathrm{Q}_n^{[\mathrm{W},\mathrm{S}]}$, and sends it to server $n$. 
Each query is a deterministic or stochastic function of $\mathrm{W}$ and $\mathrm{S}$, and hence the superscript $[\mathrm{W},\mathrm{S}]$. 
The queries may depend on the user's side information but not on any other messages that are unavailable to the user. 
Moreover, for each $n\in [N]$, the query $\mathrm{Q}_n^{[\mathrm{W},\mathrm{S}]}$ must keep $(\mathrm{W},\mathrm{S})$ private from server $n$, i.e.,
\begin{align}\label{eq:PC}
I(\mathbf{W},\mathbf{S}; \mathbf{Q}_n^{[\mathbf{W},\mathbf{S}]}) = 0.
\end{align}
This requirement is referred to as the \emph{privacy condition}.
 
Subsequently, each server $n$ generates an answer $\mathrm{A}_n^{[\mathrm{W},\mathrm{S}]}$ and sends it to the user. 
The answer is a deterministic function of the query and messages, i.e., 
${H(\mathbf{A}_n^{[\mathrm{W},\mathrm{S}]}|\mathbf{Q}_n^{[\mathrm{W},\mathrm{S}]},\mathbf{X}_{[K]})=0}$. 
Moreover, the user's demand $\mathrm{X}_{\mathrm{W}}$ must be recoverable given 
${\mathrm{A}^{[\mathrm{W},\mathrm{S}]}_{[N]}:=\{\mathrm{A}_n^{[\mathrm{W},\mathrm{S}]}\}_{n\in [N]}}$, ${\mathrm{Q}^{[\mathrm{W},\mathrm{S}]}_{[N]}:=\{\mathrm{Q}_n^{[\mathrm{W},\mathrm{S}]}\}_{n\in [N]}}$, and the user's side information $\mathrm{X}_{\mathrm{S}}$, 
i.e., 
\begin{equation}\label{eq:RC}
H(\mathbf{X}_{\mathrm{W}}| \mathbf{A}_{[N]}^{[\mathrm{W},\mathrm{S}]},\mathbf{Q}_{[N]}^{[\mathrm{W},\mathrm{S}]},\mathbf{X}_{\mathrm{S}})=0.
\end{equation}
This requirement is referred to as the \emph{recoverability condition}.

The problem is to design a scheme for generating the queries ${\mathrm{Q}_{[N]}^{[\mathrm{W},\mathrm{S}]}}$ and the corresponding answers ${\mathrm{A}_{[N]}^{[\mathrm{W},\mathrm{S}]}}$ for any given $(\mathrm{W},\mathrm{S})$ such that both the privacy and recoverability conditions defined in~\eqref{eq:PC} and~\eqref{eq:RC} are satisfied.
This problem is known as \emph{multi-server Private Information Retrieval with Private Side Information}, commonly referred to as \emph{PIR-PSI}. 

The \emph{capacity} of PIR-PSI is defined as the maximum achievable rate among all PIR-PSI schemes, where the \emph{rate} of a PIR-SI scheme is defined as the ratio of the number of bits required by the user, 
i.e., $H(\mathbf{X}_{\mathrm{W}})$, 
to the expected number of bits downloaded from all servers, 
i.e., ${\sum_{n\in [N]} H(\mathbf{A}^{[\mathrm{W},\mathrm{S}]}_n|\mathbf{Q}^{[\mathrm{W},\mathrm{S}]}_n)}$, where the expectation is taken over all possible query realizations. 

Our goal in this work is to characterize the capacity of PIR-PSI  in terms of the parameters $N$, $K$, $M$, $L$, and $q$. 

\section{Main Results}
In this section, we present our main results on the capacity of the PIR-PSI problem.  

\begin{theorem}\label{thm:1}
The capacity of PIR-PSI with $N$ servers, $K$ messages, and $M$ side information messages, where each message is composed of $L$ symbols from $\mathbbmss{F}_q$, is given by 
\begin{equation}\label{eq:R}
R := \frac{(N^{K-M}-N^{K-M-1})}{(N^{K-M}-1)}, 
\end{equation} when $N\geq M+1$ and $N-1\mid L$.
\end{theorem}

The proof of converse follows directly from the result of~\cite[Theorem~1]{CWJ2020} which states that the capacity of PIR-PSI is upper bounded by $R$, as defined in~\eqref{eq:R}, for all values of $N$, $K$, $M$, $L$, and $q$. 
To establish the achievability of the rate $R$, we propose a novel PIR-PSI scheme, 
applicable to all values of $N$, $K$, $M$, $L$, and $q$ where ${N\geq M+1}$ and ${N-1 \mid L}$. 

In the proposed scheme, each message is partitioned into ${N-1}$ sub-packets of equal size. 
Each server provides the user with the sum of a specific subset of message sub-packets,  
and the user recovers each sub-packet of the demand message from the sum retrieved from a distinct server. 
To satisfy the privacy and recoverability conditions, the scheme employs a carefully designed randomized algorithm, specifying the contributions of sub-packets from different messages to each retrieved sum. 

\begin{remark}\normalfont
As was shown in~\cite{CWJ2020}, the rate $R$ is achievable for all $N$, $K$, and $M$ when $L$ and $q$ are sufficiently large. 
Specifically, the scheme presented in~\cite{CWJ2020} requires that each message comprises $N^K$ symbols from $\mathbbmss{F}_q$, 
resulting in a sub-packetization level that grows exponentially with $K$. 
Additionally, the scheme of~\cite{CWJ2020} relies on an MDS code with its length and dimension growing exponentially with $K$. 
This implies that the field size $q$ must also grow exponentially with $K$.
In contrast, our scheme operates with a sub-packetization level of $N-1$, independent of $K$, and works over any finite field, including the binary field, without requiring an MDS code. 
\end{remark}

\begin{remark}\normalfont
For cases in which $L$ does not grow with $K$, and $q$ is as small as $2$, when $N\leq M$ or $N-1\nmid L$, 
the capacity generally remains open. 
Although details are omitted here, when $N=2$, $M=2$, $L=1$, and $q=2$, 
it can be shown that for $K=3$, 
the capacity is equal to $1$, 
which coincides with the rate $R$, 
while for $K=4$, the capacity is equal to $5/8$, 
which is less than $R = 2/3$. 
This proves that the converse bound is not tight in general. 
While the proof of achievability of these rates is based on an extension of our scheme---not presented here, a general achievability scheme remains unknown.  
\end{remark}

\begin{remark}\normalfont
While randomized schemes were previously employed in the contexts of PIR~\cite{TSC2019} and PIR with non-private side information (PIR-SI)~\cite{WHS2022,WHS2024}, 
our randomized scheme incorporates two unique features. 
In contrast to the PIR scheme presented in~\cite{TSC2019}, our scheme generates queries based on a non-uniform probability distribution, dependent on the support size of each query. 
Additionally, our scheme ensures that all possible queries of certain support sizes are generated with equal probability, a condition not required in the PIR-SI schemes presented in~\cite{WHS2022} and~\cite{WHS2024}. 
\end{remark}

\section{A New PIR-PSI Scheme}
In this section, we present a new PIR-PSI scheme, applicable to all values of $N$, $K$, $M$, $L$, and $q$ where ${N\geq M+1}$ and ${N-1\mid L}$. 
An illustrative example of the proposed scheme can be found in Appendix~\ref{sec:Example}. 

The proposed scheme operates on message sub-packets, where each message $\mathrm{X}_i$ is divided into $N-1$ sub-packets $\mathrm{X}_{i,1},\dots,\mathrm{X}_{i,N-1}$, each composed of $L/(N-1)$ symbols from $\mathbbmss{F}_q$. 
Specifically, the user's query to each server $n$ is a vector $\mathrm{v}_n\in {[0:N-1]^K}$. 
If the vector $\mathrm{v}_n$ is 
nonzero, 
server $n$ provides the user with an answer in the form of a sum that includes the message sub-packet $\mathrm{X}_{i,\mathrm{v}_n(i)}$ for each $i\in [K]$, where $\mathrm{X}_{i,0}$ is defined as a degenerate sub-packet, composed of $L/(N-1)$ zeros. 
Otherwise, if the vector $\mathrm{v}_n$ is all-zero, 
server $n$ does not send an answer back to the user. 

\vspace{0.125cm}
\textbf{Query Generation:} 
Given the demand's index $\mathrm{W}$ and the side information's index set $\mathrm{S}$, the user generates the query vectors 
$\mathrm{v}_1,\dots,\mathrm{v}_N$ using an algorithm outlined below, and 
sends each vector $\mathrm{v}_n$ to server $n$ as the query $\mathrm{Q}^{[\mathrm{W},\mathrm{S}]}_n$. 

For a randomly chosen permutation $\pi$ of $[N-1]$, 
the user constructs $N-1$ vectors $\mathrm{a}_1,\dots,\mathrm{a}_{N-1}\in {[0:N-1]^K}$ such that for each $m\in [N-1]$, $\mathrm{a}_m(\mathrm{W}) = \pi(m)$, and $\mathrm{a}_m(i) = 0$ for all $i\in [K]\setminus \{\mathrm{W}\}$. 
Each vector $\mathrm{a}_m$ corresponds to one distinct sub-packet of the demand message. 

The user then randomly generates a vector $\mathrm{b}\in [0:N-1]^{K}$ with 
${\mathrm{supp}(\mathrm{b}) = \mathrm{S}}$. 
The vector $\mathrm{b}$ corresponds to a sum of $M$ sub-packets, each belonging to one side information message.

Next, the user randomly selects a pair $(I,J)$ from the set ${[0:M]\times [0:K-M-1]}$, where the probability of selecting each pair $(i,j)$ is equal to $P_{i,j}$ defined as
\begin{equation}\label{eq:Pij}
P_{i,j} := \binom{M}{i}\binom{K-M-1}{j} \frac{m_{i,j}}{N^{K-M-1}}
\end{equation} 
for all $0\leq i\leq M$ and $0\leq j\leq K-M-1$, 
where $m_{0,0} := 1$, 
$m_{i,j} :=0$ if $1\leq i+j\leq M$, 
and 
\begin{equation}\label{eq:mij}
m_{i,j} := \sum_{k=0}^{i+j-M-1}(-1)^k \binom{M+k-1}{k}(N-1)^{i+j-M-k}
\end{equation}
if $M+1\leq i+j\leq K-1$. 
It should be noted that $P_{i,j}$'s form a probability distribution, as stated in the following lemma. 
The proof of this lemma can be found in Appendix~\ref{sec:ProofsLem1Lem4}.  

\begin{lemma}\label{lem:Pij}
For $P_{i,j}$'s defined in~\eqref{eq:Pij}, 
$P_{i,j}\geq 0$ for all $0\leq i\leq M$ and $0\leq j\leq K-M-1$, 
and 
$\sum_{i=0}^{M}\sum_{j=0}^{K-M-1} P_{i,j} = 1$. 
\end{lemma}

For a randomly chosen $I$-subset ${\mathrm{R}\subseteq \mathrm{S}}$, the user forms 
a vector $\mathrm{b}_0\in [0:N-1]^K$ with 
$\mathrm{b}_0(i) = \mathrm{b}(i)$ for all $i\in \mathrm{R}$ and $\mathrm{b}_0(i)=0$ otherwise. 
If $I=0$, the vector  $\mathrm{b}_0$ is all-zero;
otherwise, 
$\mathrm{b}_0$ corresponds to a sum of $I$ sub-packets, each belonging to one side information message.
If $I\neq 0$, 
the user also forms 
$I$ vectors $\mathrm{b}_1,\dots,\mathrm{b}_I\in [0:N-1]^K$ such that for each $m\in [I]$,  
$\mathrm{b}_m(i) = \mathrm{b}_0(i)$ for all $i\in \mathrm{R}_m$ and $\mathrm{b}_m(i)=0$ otherwise, 
where  
$\mathrm{R}_1,\dots,\mathrm{R}_I$ are randomly ordered $(I-1)$-subsets of $\mathrm{R}$. 
Each vector $\mathrm{b}_m$  corresponds to a sum of $I-1$ sub-packets, each belonging to one side information message.  

Then, for a randomly chosen $J$-subset ${\mathrm{T}\subseteq [K]\setminus \{\mathrm{W}\}\cup \mathrm{S}}$, 
the user randomly generates a vector ${\mathrm{c}\in [0:N-1]^{K}}$ with ${\mathrm{supp}(\mathrm{c}) = \mathrm{T}}$. 
If $J=0$, 
the vector $\mathrm{c}$ is all-zero; 
otherwise, $\mathrm{c}$ corresponds to a sum of $J$ sub-packets, each associated with one interference message.  

Next, the user randomly selects ${\theta\in \{0,1\}}$, 
where the probability of selecting $0$ is equal to $P_I$ defined as 
\begin{equation}\label{eq:PI}
P_I := \frac{M-I}{M-I+1}.
\end{equation} 

Given $\{\mathrm{a}_m\}_{m\in [N-1]}$, $\mathrm{b}$, $\mathrm{b}_0$, $\{\mathrm{b}_m\}_{m\in [I]}$, $\mathrm{c}$, and $\theta$, 
the user forms
$N$ vectors $\mathrm{u}_1,\dots,\mathrm{u}_N\in [0:N-1]^K$ as follows: 

\vspace{0.125cm}
\noindent If $\theta=0$:  
\[
\begin{array}{ll}
\mathrm{u}_{1} = {\mathrm{b}_{0}+\mathrm{c}} \\
\mathrm{u}_{m+1} = \mathrm{a}_{m}+\mathrm{b}+\mathrm{c} & \quad \forall m\in [N-1];
\end{array}
\]
Otherwise: 
\[
\begin{array}{ll}
\mathrm{u}_{1} = {\mathrm{b}_{0}+\mathrm{c}}\\
\mathrm{u}_{m+1} = {\mathrm{a}_{m}+\mathrm{b}_m+\mathrm{c}} & \quad \forall m\in [I],\\
\mathrm{u}_{m+1} = \mathrm{a}_{m}+\mathrm{b}+\mathrm{c} & \quad \forall m\in [I+1:N-1].
\end{array}
\]

Given the vectors $\{\mathrm{u}_m\}_{m\in [N]}$ and a randomly chosen permutation $\sigma$ of $[N]$, 
the user then constructs the vectors $\{\mathrm{v}_n\}_{n\in [N]}$ by setting $\mathrm{v}_n = \mathrm{u}_{\sigma(n)}$ for all $n\in [N]$. 

\begin{lemma}\label{lem:PC}
Every query $\mathrm{Q}_n^{[\mathrm{W},\mathrm{S}]}$ generated by the proposed scheme satisfies the privacy condition. 
\end{lemma}

\begin{proof}
The proof is given in Section~\ref{subsec:PC}.
\end{proof}

\textbf{Answer Generation:} 
Given the query $\mathrm{Q}^{[\mathrm{W},\mathrm{S}]}_n= \mathrm{v}_n$, 
if the vector $\mathrm{v}_n$ is nonzero, 
server $n$ generates the sum of message sub-packets corresponding to the vector $\mathrm{v}_n$, 
i.e., 
$\mathrm{Y}_n = \sum_{i=1}^{K} \mathrm{X}_{i,\mathrm{v}_n(i)}$, 
and sends $\mathrm{Y}_n$ to the user as the answer $\mathrm{A}^{[\mathrm{W},\mathrm{S}]}_n$. 
Otherwise, if the vector $\mathrm{v}_n$ is all-zero, 
server $n$ does not provide the user with an answer, and 
the user sets $\mathrm{Y}_n = 0$. 

\begin{lemma}\label{lem:RC}
Every collection of answers $\{\mathrm{A}^{[\mathrm{W},\mathrm{S}]}_n\}$ generated by the proposed scheme satisfies the recoverability condition. 
\end{lemma}

\begin{proof}
The proof is given in Section~\ref{subsec:RC}.
\end{proof}

\textbf{Rate Calculation:} 
By definition, the rate is the ratio of $H(\mathbf{X}_{\mathrm{W}})$ to ${\sum_{n\in [N]} H(\mathbf{A}^{[\mathrm{W},\mathrm{S}]}_n|\mathbf{Q}^{[\mathrm{W},\mathrm{S}]}_n)}$. 
It is clear that $H(\mathbf{X}_{\mathrm{W}}) = L$ in the $q$-ary unit, 
because each message $\mathbf{X}_i$ is uniformly distributed over $\mathbbmss{F}_q^{L}$. 
We also observe that 
$\sum_{n\in [N]} H(\mathbf{A}^{[\mathrm{W},\mathrm{S}]}_n|\mathbf{Q}^{[\mathrm{W},\mathrm{S}]}_n) = \sum_{n\in [N]} H(\mathbf{Y}_n|\mathbf{v}_n)$. 
According to the answer generation algorithm used in the proposed scheme, 
if the vector $\mathbf{v}_n$ is all-zero, 
the answer $\mathbf{Y}_n$ is $0$, and 
if $\mathbf{v}_n$ is nonzero, 
$\mathbf{Y}_n$ is the sum of the message sub-packets  $\{\mathbf{X}_{i,\mathbf{v}_n(i)}\}_{i\in [K]}$. 
It is obvious that
${H(\mathbf{Y}_n|\mathbf{v}_n=0)} = 0$. 
Additionally, 
${H(\mathbf{Y}_n|\mathbf{v}_n\neq 0)} = L/(N-1)$. 
This is evident as $\mathbf{Y}_n$ is the sum of a (non-empty) set of non-degenerate message sub-packets---independent and uniformly distributed over $\mathbbmss{F}^{L/(N-1)}_q$.
Thus, 
$H(\mathbf{Y}_n|\mathbf{v}_n) = {\mathbb{P}(\mathbf{v}_n \neq 0)L/(N-1)}$.
Moreover, 
$\mathbb{P}(\mathbf{v}_n \neq 0) = 1-\mathbb{P}(\mathbf{v}_n =0) = 1-P_{0,0}/N =1-1/N^{K-M}$ (see Section~\ref{subsubsec:s0}).
By combining these, 
$H(\mathbf{Y}_n|\mathbf{v}_n) = (1-1/N^{K-M}) L/(N-1)$, 
which implies that 
$\sum_{n\in [N]} H(\mathbf{Y}_n|\mathbf{v}_n) = N(1-1/N^{K-M})L/(N-1)$, noting that $H(\mathbf{Y}_n|\mathbf{v}_n)$ remains constant for all ${n\in [N]}$. 
Thus, the rate of the proposed scheme is equal to ${(N-1)/(N-1/N^{K-M-1})}$, matching the rate $R$ in~\eqref{eq:R}. 

\section{Proofs of Lemmas~\ref{lem:PC} and~\ref{lem:RC}}\label{sec:Proofs}
We present the proof of Lemma~\ref{lem:PC} in Section~\ref{subsec:PC} and the proof of Lemma~\ref{lem:RC} in Section~\ref{subsec:RC}.  

\subsection{Proof of Lemma~\ref{lem:PC}}\label{subsec:PC}
The following lemma (Lemma~\ref{lem:Summij}) is useful in proving Lemma~\ref{lem:PC}, and its proof can be found in Appendix~\ref{sec:ProofsLem1Lem4}. 

\begin{lemma}\label{lem:Summij}
For any $1\leq j\leq K-M-1$,  
\[\sum_{i = (M-j+1)^{+}}^{M} \left(\binom{M}{i} (N-1) - \binom{M}{i-1} \right) m_{i,j} = m_{M,j+1},\]
where $m_{i,j}$'s are defined in~\eqref{eq:mij}. 
\end{lemma}

To satisfy the privacy condition, 
it is required that, 
for any query vector $\mathrm{v}$ generated by the proposed scheme and any server $n$, ${\mathbb{P}(\mathbf{Q}^{[\mathbf{W},\mathbf{S}]}_n=\mathrm{v}|\mathbf{W}=\mathrm{W},\mathbf{S} = \mathrm{S})}$, which we denote by
${\mathbb{P}(\mathbf{Q}^{[\mathrm{W},\mathrm{S}]}_n=\mathrm{v})}$ for brevity, 
remains constant 
for all ${\mathrm{W}\in [K]}$ and all $M$-subsets ${\mathrm{S}\subseteq [K]\setminus \{\mathrm{W}\}}$.  

Fix an arbitrary query vector $\mathrm{v}$ with a support size $s$. 
According to the query generation algorithm used in the proposed scheme, 
either $s=0$ or $M+1\leq s\leq K$. 
In the following, we analyze the cases of $s=0$, $s=M+1$, $M+1<s<K$, and $s=K$ separately.
In each case, we consider an arbitrary pair $(\mathrm{W},\mathrm{S})$, and 
show that ${\mathbb{P}(\mathbf{Q}^{[\mathrm{W},\mathrm{S}]}_n=\mathrm{v})}$ does not depend on $\mathrm{W}$ and $\mathrm{S}$. 
For simplicity, we denote ${|\mathrm{S}\cap \mathrm{supp}(\mathrm{v})|}$ by $r$ 
without explicitly referring to $\mathrm{S}$.

\vspace{0.125cm}
\subsubsection{$s=0$}\label{subsubsec:s0} 
Since $\mathrm{v}$ is all-zero in this case, it is immediate that 
$\mathbb{P}(\mathbf{Q}^{[\mathrm{W},\mathrm{S}]}_n=\mathrm{v}) = P_{0,0}\times 1/N = 1/{N^{K-M}}$, 
independent of $\mathrm{W}$ and $\mathrm{S}$. 
This is because $\mathrm{v}$ is all-zero only if $(\boldsymbol{I},\boldsymbol{J})=(0,0)$, which occurs with probability $P_{0,0} = 1/{N^{K-M-1}}$, and $\boldsymbol{\sigma}(1) = n$, which occurs with probability $1/N$. 

\vspace{0.125cm}
\subsubsection{$s=M+1$}
In this case, we need to consider two scenarios: 
$\mathrm{W}\not\in \mathrm{supp}(\mathrm{v})$ and $\mathrm{W}\in \mathrm{supp}(\mathrm{v})$. 

If $\mathrm{W}\not\in \mathrm{supp}(\mathrm{v})$, 
we have
${(2M+2-K)^{+}\leq r\leq M}$.
Since the vector $\mathrm{u}_1$ is the only one among     $\mathrm{u}_1,\dots,\mathrm{u}_{N}$ whose support does not contain $\mathrm{W}$, 
it readily follows that the event $\mathbf{Q}^{[\mathrm{W},\mathrm{S}]}_n=\mathrm{v}$ is equivalent to the set of events $(\boldsymbol{I},\boldsymbol{J}) = (r,M+1-r)$, 
$\boldsymbol{\theta}\in \{0,1\}$, 
$\mathbf{u}_1 = \mathrm{v}$, and $\boldsymbol{\sigma}(1) = n$. 
Since $\mathrm{u}_1 = \mathrm{b}_0+\mathrm{c}$, 
it follows that the event 
$\mathbf{u}_1 = \mathrm{v}$ is equivalent to the set of events 
(i) $\mathbf{b}_0(k) = \mathrm{v}(k)$ for all $k\in \mathrm{S}\cap \mathrm{supp}(\mathrm{v})$, 
or equivalently, 
$\mathbf{b}(k) = \mathrm{v}(k)$ for all $k\in \mathrm{S}\cap \mathrm{supp}(\mathrm{v})$ and  
$\mathbf{R} = \mathrm{S}\cap \mathrm{supp}(\mathrm{v})$, 
and
(ii) $\mathbf{c}(k) = \mathrm{v}(k)$ for all $k\in  \mathrm{supp}(\mathrm{v})\setminus \mathrm{S}$.
Thus, 
$\mathbb{P}(\mathbf{Q}^{[\mathrm{W},\mathrm{S}]}_n=\mathrm{v})$ is given by  
\begin{align}\label{eq:C1P1}
& P_{r,M+1-r} \times \frac{1}{(N-1)^{r}} \times \frac{1}{\binom{M}{r}} \nonumber \\ 
& \quad \times \frac{1}{(N-1)^{M+1-r}}\times \frac{1}{\binom{K-M-1}{M+1-r}}\times \frac{1}{N}.
\end{align} 

If $\mathrm{W}\in \mathrm{supp}(\mathrm{v})$, 
we observe that ${0\leq r\leq M}$. 
We analyze the cases of $0\leq r\leq M-1$ and $r=M$ separately. 

For each $0\leq r\leq M-1$, 
the event  
$\mathbf{Q}^{[\mathrm{W},\mathrm{S}]}_n = \mathrm{v}$ is equivalent to the set of events 
$(\boldsymbol{I},\boldsymbol{J}) = (r+1,M-r)$, 
$\boldsymbol{\theta}=1$, 
$\mathbf{u}_{m+1} = \mathrm{v}$, and 
$\boldsymbol{\sigma}(m+1) = n$ for some ${m\in [r+1]}$. 
Since $I = r+1\neq 0$, we observe that $\mathrm{u}_{m+1} = \mathrm{a}_{m} + \mathrm{b}_{m}+\mathrm{c}$ for all $m\in [r+1]$, and hence,  
for each $m\in [r+1]$,  
the event $\mathbf{u}_{m+1} = \mathrm{v}$ is equivalent to the following set of events: 
$\mathbf{a}_{m}(\mathrm{W}) = \mathrm{v}(\mathrm{W})$, or equivalently, $\boldsymbol{\pi}(m) = \mathrm{v}(\mathrm{W})$; 
${\mathbf{b}_{m}(k) = \mathrm{v}(k)}$ for all $k\in \mathrm{S}\cap \mathrm{supp}(\mathrm{v})$,
or equivalently, ${\mathbf{b}(k) = \mathrm{v}(k)}$ for all $k\in \mathrm{S}\cap \mathrm{supp}(\mathrm{v})$, 
$\mathbf{R} = \mathbf{R}_m \cup \{k\}$ for some $k\in \mathrm{S}\setminus \mathrm{supp}(\mathrm{v})$, and $\mathbf{R}_m = \mathrm{S}\cap \mathrm{supp}(\mathrm{v})$; 
and 
${\mathbf{c}(k) = \mathrm{v}(k)}$ for all ${k\in  \mathrm{supp}(\mathrm{v})\setminus \{\mathrm{W}\}\cup \mathrm{S}}$.
Putting these together, 
it follows that $\mathbb{P}(\mathbf{Q}^{[\mathrm{W},\mathrm{S}]}_n=\mathrm{v})$ is given by  
\begin{align}\label{eq:C1P2}
& P_{r+1,M-r} \times \frac{1}{M-r} \times (r+1) \times \frac{1}{N-1} \times \frac{1}{(N-1)^{r}} \nonumber \\ & \quad \times \frac{M-r}{\binom{M}{r+1}} \times \frac{1}{r+1} \times \frac{1}{(N-1)^{M-r}}\times \frac{1}{\binom{K-M-1}{M-r}}\times \frac{1}{N}.
\end{align} 

For $r=M$, it follows that  $\mathrm{supp}(
\mathrm{v})=\{\mathrm{W}\}\cup\mathrm{S}$, 
which readily implies that the event 
$\mathbf{Q}^{[\mathrm{W},\mathrm{S}]}_n = \mathrm{v}$ is equivalent to the set of events $(\boldsymbol{I},\boldsymbol{J}) = (0,0)$, 
$\boldsymbol{\theta}\in \{0,1\}$, 
$\mathbf{u}_{m+1} = \mathrm{v}$, and 
$\boldsymbol{\sigma}(m+1) = n$ for some ${m\in [N-1]}$.
Since $I=0$, it follows that  $\mathrm{u}_{m+1} = \mathrm{a}_{m} + \mathrm{b}+\mathrm{c}$ for all $m\in [N-1]$, which implies that, 
for each $m\in [N-1]$, 
the event $\mathbf{u}_{m+1} = \mathrm{v}$ is equivalent to the set of events  
$\mathbf{a}_{m}(\mathrm{W}) = \mathrm{v}(\mathrm{W})$, 
and 
$\mathbf{b}(k) = \mathrm{v}(k)$ for all $k\in \mathrm{S}$.
Thus, $\mathbb{P}(\mathbf{Q}^{[\mathrm{W},\mathrm{S}]}_n=\mathrm{v})$ is given by  
\begin{align}\label{eq:C1P3}
& P_{0,0} \times \frac{1}{N-1} \times \frac{1}{(N-1)^{M}} \times \frac{1}{N}.
\end{align}

To complete the proof, 
it suffices to show the equality of~\eqref{eq:C1P1},~\eqref{eq:C1P2}, and~\eqref{eq:C1P3}.
By applying~\eqref{eq:Pij}, it is easy to see that \eqref{eq:C1P1}-\eqref{eq:C1P3} can be simplified to $m_{r,M+1-r}/(N^{K-M}(N-1)^{M+1})$,
${m_{r+1,M-r}/(N^{K-M}(N-1)^{M+1})}$, and $1/(N^{K-M}(N-1)^{M})$, respectively. 
According to~\eqref{eq:mij}, it follows that 
${m_{i,j} = N-1}$ for all ${0\leq i\leq M}$ and ${0\leq j\leq K-M-1}$ such that ${i+j = M+1}$, 
which implies that 
$m_{r,M+1-r} = N-1$ for all $(2M+2-K)^{+}\leq r\leq M$, and 
$m_{r+1,M-r} = N-1$ for all $0\leq r\leq M-1$. 
Substituting these into the preceding expressions for~\eqref{eq:C1P1} and~\eqref{eq:C1P2} yields the equality of~\eqref{eq:C1P1}-\eqref{eq:C1P3}.   

\vspace{0.125cm}
\subsubsection{$M+1<s<K$}
Similar to the case of $s=M+1$, we consider two scenarios: 
${\mathrm{W}\not\in \mathrm{supp}(\mathrm{v})}$ and ${\mathrm{W}\in \mathrm{supp}(\mathrm{v})}$. 

If ${\mathrm{W}\not\in \mathrm{supp}(\mathrm{v})}$, 
we have
${(s-K+M+1)^{+}\leq r\leq M}$. 
Similarly as in the case of $s=M+1$, 
it can be shown that the event $\mathbf{Q}^{[\mathrm{W},\mathrm{S}]}_n=\mathrm{v}$ is equivalent to the set of events $(\boldsymbol{I},\boldsymbol{J}) = (r,s-r)$, 
$\boldsymbol{\theta}\in \{0,1\}$, 
$\mathbf{b}(k) = \mathrm{v}(k)$ for all $k\in \mathrm{S}\cap \mathrm{supp}(\mathrm{v})$, 
$\mathbf{R} = \mathrm{S}\cap \mathrm{supp}(\mathrm{v})$, 
$\mathbf{c}(k) = \mathrm{v}(k)$ for all $k\in \mathrm{supp}(\mathrm{v})\setminus \mathrm{S}$, and $\boldsymbol{\sigma}(1) = n$. 
Thus, 
$\mathbb{P}(\mathbf{Q}^{[\mathrm{W},\mathrm{S}]}_n=\mathrm{v})$ is given by  
\begin{align}\label{eq:C2P1}
& P_{r,s-r} \times \frac{1}{(N-1)^{r}} \times \frac{1}{\binom{M}{r}} \nonumber \\ 
& \quad \times \frac{1}{(N-1)^{s-r}}\times \frac{1}{\binom{K-M-1}{s-r}}\times \frac{1}{N}.
\end{align} 

If $\mathrm{W}\in \mathrm{supp}(\mathrm{v})$, we have ${(s-K+M)^{+}\leq r\leq M}$. 
Similar to the case of $s=M+1$, 
we analyze the cases of $(s-K+M)^{+}\leq r\leq M-1$ and $r=M$ separately. 

In the case of $(s-K+M)^{+}\leq r\leq M-1$, 
using the same arguments as in the case of ${s=M+1}$ and $0\leq r\leq M-1$, 
it can be shown that the event 
$\mathbf{Q}^{[\mathrm{W},\mathrm{S}]}_n = \mathrm{v}$ is equivalent to the set of events $(\boldsymbol{I},\boldsymbol{J}) = {(r+1,s-r-1)}$, 
$\boldsymbol{\theta}=1$, 
$\mathbf{a}_{m}(\mathrm{W}) = \mathrm{v}(\mathrm{W})$, 
${\mathbf{b}(k) = \mathrm{v}(k)}$ for all $k\in {\mathrm{S}\cap \mathrm{supp}(\mathrm{v})}$, 
$\mathbf{R} = \mathbf{R}_m \cup \{k\}$ for some $k\in \mathrm{S}\setminus \mathrm{supp}(\mathrm{v})$, 
$\mathbf{R}_m = \mathrm{S}\cap \mathrm{supp}(\mathrm{v})$, ${\mathbf{c}(k) = \mathrm{v}(k)}$ for all ${k\in  \mathrm{supp}(\mathrm{v})\setminus \{\mathrm{W}\}\cup \mathrm{S}}$, and 
$\boldsymbol{\sigma}(m+1) = n$, for some ${m\in [r+1]}$. 
Thus, $\mathbb{P}(\mathbf{Q}^{[\mathrm{W},\mathrm{S}]}_n=\mathrm{v})$ is given by  
\begin{align}\label{eq:C2P2}
& P_{r+1,s-r-1} \times \frac{1}{M-r} \times (r+1) \times \frac{1}{N-1} \times \frac{1}{(N-1)^{r}} \nonumber \\ & \quad \times \frac{M-r}{\binom{M}{r+1}} \times \frac{1}{r+1} \times \frac{1}{(N-1)^{s-r-1}}\times \frac{1}{\binom{K-M-1}{s-r-1}}\times \frac{1}{N}.
\end{align}

In the case of $r = M$, however, 
the analysis differs from that in the case of $s=M+1$ and $r=M$. 
In this case, 
there are two distinct event sets that can yield the event $\mathbf{Q}^{[\mathrm{W},\mathrm{S}]}_n = \mathrm{v}$. 

Note that $\mathrm{supp}(\mathrm{v}) = \{\mathrm{W}\}\cup \mathrm{S} \cup \mathrm{T}$, 
where $\mathrm{T}$ is a ${j:=s-M-1}$-subset of $[K]\setminus \{\mathrm{W}\}\cup \mathrm{S}$. 
Also, note that 
$M+1\leq i+j\leq K-M-1$ for any $(M-j+1)^{+}\leq i\leq M$. 
By using these results, 
it is easy to show that, 
for every $(M-j+1)^{+}\leq i\leq M$, the event $\mathbf{Q}^{[\mathrm{W},\mathrm{S}]}_n = \mathrm{v}$ occurs if and only if one of the following two sets of events occurs:\vspace{0.125cm} 

(i) $(\boldsymbol{I},\boldsymbol{J}) = (i,j)$, 
$\boldsymbol{\theta}=0$, 
$\mathbf{a}_{m}(\mathrm{W}) = \mathrm{v}(\mathrm{W})$, 
$\mathbf{b}(k) = \mathrm{v}(k)$ for all $k\in \mathrm{S}$, 
${\mathbf{c}(k) = \mathrm{v}(k)}$ for all ${k\in  \mathrm{supp}(\mathrm{v})\setminus \{\mathrm{W}\}\cup \mathrm{S}}$,
and 
$\boldsymbol{\sigma}(m+1) = n$, for some ${m\in [N-1]}$;\vspace{0.125cm}   

(ii) $(\boldsymbol{I},\boldsymbol{J}) = (i,j)$, 
$\boldsymbol{\theta}=1$, 
$\mathbf{a}_{m}(\mathrm{W}) = \mathrm{v}(\mathrm{W})$, 
$\mathbf{b}(k) = \mathrm{v}(k)$ for all $k\in \mathrm{S}$, 
${\mathbf{c}(k) = \mathrm{v}(k)}$ for all ${k\in  \mathrm{supp}(\mathrm{v})\setminus \{\mathrm{W}\}\cup \mathrm{S}}$,
and 
$\boldsymbol{\sigma}(m+1) = n$, for some ${m\in [i+1:N-1]}$.\vspace{0.125cm} 

Taking into account both event sets (i) and (ii), 
it becomes evident that  
$\mathbb{P}(\mathbf{Q}^{[\mathrm{W},\mathrm{S}]}_n=\mathrm{v})$ is given by
\begin{align}\label{eq:C2P3}
& \sum_{i=(M-j+1)^{+}}^{M} P_{i,j} \times \Big((N-1) P_i + (N-i-1) (1-P_{i})\Big) \nonumber \\
& \quad \times \frac{1}{N-1} \times \frac{1}{(N-1)^{M}} \times \frac{1}{(N-1)^{j}}\times \frac{1}{\binom{K-M-1}{j}}\times \frac{1}{N},
\end{align} 
for any $M+1<s<K$, where $j=s-M-1$. 

We need to show that~\eqref{eq:C2P1}-\eqref{eq:C2P3} are equal.  
By using~\eqref{eq:Pij}, 
we can rewrite~\eqref{eq:C2P1} and \eqref{eq:C2P2} as 
$m_{r,s-r}/(N^{K-M}(N-1)^s)$ and  
$m_{r+1,s-r-1}/(N^{K-M}(N-1)^s)$, respectively. 
Similar to the case of $s=M+1$, by using~\eqref{eq:mij}, 
it is immediate that 
$m_{r,s-r} = m_{M,s-M}$ for all $(s-K+M+1)^{+}\leq r\leq M$, 
noting that $r+(s-r) = s$ is fixed.
Similarly, 
$m_{r+1,s-r-1} = m_{M,s-M}$ for all $(s-K+M)^{+}\leq r\leq M-1$, noting that $(r+1)+(s-r-1) = s$ is fixed. 
This proves the equality of~\eqref{eq:C2P1} and~\eqref{eq:C2P2}. 
Since both~\eqref{eq:C2P1} and~\eqref{eq:C2P2} are equal to $m_{M,s-M}/(N^{K-M}(N-1)^s)$, 
it remains to 
show that~\eqref{eq:C2P3} is equal to $m_{M,s-M}/(N^{K-M}(N-1)^s)$.

By using~\eqref{eq:Pij} and~\eqref{eq:PI}, 
we can rewrite~\eqref{eq:C2P3} as 
\begin{equation}\label{eq:C2P3I1}
\sum_{i = (M-j+1)^{+}}^{M} \left(\binom{M}{i} (N-1) - \binom{M}{i-1} \right) \frac{m_{i,j}}{N^{K-M} (N-1)^{s}},
\end{equation}
noting that
\begin{dmath*}
\binom{M}{i}\Big((N-1)P_i+(N-i-1)(1-P_i)\Big) = \binom{M}{i}(N-1)-\binom{M}{i-1}
\end{dmath*} 
for all $0\leq i\leq M$.
By Lemma~\ref{lem:Summij}, 
\begin{equation*}
\sum_{i = (M-j+1)^{+}}^{M} \left(\binom{M}{i} (N-1) - \binom{M}{i-1} \right) m_{i,j} = m_{M,j+1}
\end{equation*} 
for any $1\leq j\leq K-M-1$.
By applying this result and replacing $j$ with $s-M-1$, 
it follows that~\eqref{eq:C2P3I1} is equal to $m_{M,s-M}/(N^{K-M}(N-1)^s)$ for any $M+1<s<K$, which was to be shown.  

\vspace{0.125cm}
\subsubsection{$s=K$}
In this case, $\mathrm{supp}(\mathrm{v}) = \{\mathrm{W}\}\cup \mathrm{S}\cup \mathrm{T}$, where $\mathrm{T} = [K]\setminus \{\mathrm{W}\}\cup \mathrm{S}$. 
Setting $j=K-M-1$ and using the same analysis as in the case with $M+1<s<K$ and $r=M$, 
it is easy to verify that, 
for every ${(M+1-j)^{+}\leq i\leq M}$, 
the event $\mathbf{Q}^{[\mathrm{W},\mathrm{S}]}_n = \mathrm{v}$ occurs if and only if one of the previously specified event sets (i) or (ii) occurs.
Accordingly, by substituting $j$ with ${K-M-1}$ and $s$ with $K$ in~\eqref{eq:C2P3I1}, 
it follows that 
$\mathbb{P}(\mathbf{Q}^{[\mathrm{W},\mathrm{S}]}_n=\mathrm{v})$ is given by 
\begin{align*}\label{eq:C3}
&\sum_{i = (2M-K+2)^{+}}^{M} \hspace{-0.175cm}\left(\binom{M}{i} (N-1) - \binom{M}{i-1} \right)\frac{m_{i,K-M-1}}{N^{K-M} (N-1)^{K}},
\end{align*} 
which remains constant, regardless of $\mathrm{W}$ and $\mathrm{S}$. 

\subsection{Proof of Lemma~\ref{lem:RC}}\label{subsec:RC}
We need to show that all the $N-1$ (non-degenerate) sub-packets of the demand message, 
$\{\mathrm{X}_{\mathrm{W},n}\}_{n\in [N-1]}$, can be recovered, given the answers $\mathrm{A}^{[\mathrm{W},\mathrm{S}]}_n = \mathrm{Y}_n$ for $n\in [N]$.  
We consider the cases of $\theta=0$ and $\theta=1$ separately. 

Let $\tau$ be the inverse of the permutation $\sigma$, 
i.e., $\mathrm{u}_n = \mathrm{v}_{\tau(n)}$ for all $n\in [N]$, and 
let $\mathrm{Z}_n:= \mathrm{Y}_{\tau(n)}$ for all $n\in [N]$. 

In the case of $\theta = 0$, 
we have 
\[\mathrm{a}_{n} = \mathrm{u}_{n+1}-\mathrm{u}_1 -\mathrm{b}+\mathrm{b}_0\] for all $n\in [N-1]$, which implies that 
\[\mathrm{X}_{\mathrm{W},\pi(n)}=\mathrm{Z}_{n+1}-\mathrm{Z}_{1}-\sum_{i=1}^{K} \mathrm{X}_{i,\mathrm{b}(i)}+\sum_{i=1}^{K} \mathrm{X}_{i,\mathrm{b}_0(i)}\] for all $n\in [N-1]$.
Note that $\sum_{i=1}^{K} \mathrm{X}_{i,\mathrm{b}(i)}$ and $\sum_{i=1}^{K} \mathrm{X}_{i,\mathrm{b}_0(i)}$ are known by the user. 
This is simply because $\mathrm{supp}(\mathrm{b})=\mathrm{S}$ and $\mathrm{supp}(\mathrm{b}_0) \subseteq \mathrm{S}$ by definition, and the messages $\mathrm{X}_{\mathrm{S}}$ form the user's side information. 
Thus, the user can recover the demand message sub-packets $\{\mathrm{X}_{\mathrm{W},\pi(n)}\}_{n\in [N-1]}$, 
given the servers' answers $\{\mathrm{Z}_n\}_{n\in [N]}$. 

In the case of $\theta = 1$, 
we have \[\mathrm{a}_{n} = \mathrm{u}_{n+1}-\mathrm{u}_1 -\mathrm{b}_n+\mathrm{b}_0\] for all $n\in [I]$, implying that 
\[\mathrm{X}_{\mathrm{W},\pi(n)}=\mathrm{Z}_{n+1}-\mathrm{Z}_{1}-\sum_{i=1}^{K} \mathrm{X}_{i,\mathrm{b}_n(i)}+\sum_{i=1}^{K} \mathrm{X}_{i,\mathrm{b}_0(i)}\] for all $n\in [I]$.  
Similar to the case of $\theta = 0$, 
it follows that $\sum_{i=1}^{K} \mathrm{X}_{i,\mathrm{b}_n(i)}$ and $\sum_{i=1}^{K} \mathrm{X}_{i,\mathrm{b}_0(i)}$ are known by the user since $\mathrm{supp}(\mathrm{b}_0) \subseteq \mathrm{S}$ and $\mathrm{supp}(\mathrm{b}_n)\subset \mathrm{S}$ for all $n\in [I]$, and the messages $\mathrm{X}_{\mathrm{S}}$ form the user's side information. 
Thus, the user can recover the sub-packets $\{\mathrm{X}_{\mathrm{W},\pi(n)}\}_{n\in [I]}$, given the servers' answers $\{\mathrm{Z}_n\}_{n\in [N]}$. 
The process of recovering the sub-packets $\{\mathrm{X}_{\mathrm{W},\pi(n)}\}_{n\in [I+1:N-1]}$ remains the same as that in the case of $\theta = 0$ and is omitted to avoid repetition. 

\bibliographystyle{IEEEtran}
\bibliography{PIR_PC_Refs}

\begin{appendices}

\section{An Illustrative Example}\label{sec:Example}
In this section, we present an illustrative example of the proposed scheme. 
Consider $N=4$ servers, storing identical copies of ${K=5}$ messages, where each message consists of ${N-1=3}$ sub-packets of equal size. 
Consider a user who knows $M=2$ of these messages  
and wants to retrieve another message from the message set. 
For ease of notation, 
we denote the demand message by $\mathrm{A}$, 
the two side information messages by $\mathrm{B}$ and $\mathrm{C}$, and 
the two interference messages by $\mathrm{D}$ and $\mathrm{E}$. 

{\renewcommand{\arraystretch}{1.85}
\begin{table*}[!t]
    \centering
    \caption{Different Types of Answer Sets and Their Corresponding Probabilities}\label{tab:1}
    \scalebox{0.96}{
    \begin{tabular}{|c|c|c|c|c|c|}
    \hline
    Type & $\mathrm{Z}_{1}$ & $\mathrm{Z}_{2}$ & $\mathrm{Z}_{3}$  & $\mathrm{Z}_{3}$  &  $P_{\{\mathrm{Z}_1,\mathrm{Z}_2,\mathrm{Z}_3,\mathrm{Z}_4\}}$ \\  
    \hline
    0 & $0$ & ${\color{black}\mathrm{A}_{\clubsuit}}+{\color{black}\mathrm{B}_i}+{\color{black}\mathrm{C}_j}$ & ${\color{black}\mathrm{A}_{\diamondsuit}}+{\color{black}\mathrm{B}_i}+{\color{black}\mathrm{C}_j}$ & ${\color{black}\mathrm{A}_{\spadesuit}}+{\color{black}\mathrm{B}_i}+{\color{black}\mathrm{C}_j}$ & $P_{0,0}\times (\frac{1}{3})^2\times \frac{1}{6} = \frac{1}{864}$ \\
    \hline
    1 & ${\color{black}\mathrm{B}_i}+\mathrm{D}_k+\mathrm{E}_l$ & ${\color{black}\mathrm{A}_{\clubsuit}}+{\color{black}\mathrm{B}_i}+{\color{black}\mathrm{C}_j}+\mathrm{D}_k+\mathrm{E}_l$ & ${\color{black}\mathrm{A}_{\diamondsuit}}+{\color{black}\mathrm{B}_i}+{\color{black}\mathrm{C}_j}+\mathrm{D}_k+\mathrm{E}_l$ & ${\color{black}\mathrm{A}_{\spadesuit}}+{\color{black}\mathrm{B}_i}+{\color{black}\mathrm{C}_j}+\mathrm{D}_k+\mathrm{E}_l$ & $P_{1,2}\times P_1 \times \frac{1}{2} \times (\frac{1}{3})^4\times \frac{1}{6} = \frac{1}{5184}$ \\
    \hline
    2 & ${\color{black}\mathrm{B}_i}+\mathrm{D}_k+\mathrm{E}_l$ & ${\color{black}\mathrm{A}_{\clubsuit}}+\mathrm{D}_k+\mathrm{E}_l$ & ${\color{black}\mathrm{A}_{\diamondsuit}}+{\color{black}\mathrm{B}_i}+{\color{black}\mathrm{C}_j}+\mathrm{D}_k+\mathrm{E}_l$ & ${\color{black}\mathrm{A}_{\spadesuit}}+{\color{black}\mathrm{B}_i}+{\color{black}\mathrm{C}_j}+\mathrm{D}_k+\mathrm{E}_l$ & $P_{1,2}\times (1-P_1) \times \frac{1}{2} \times (\frac{1}{3})^4\times \frac{1}{6} = \frac{1}{5184}$ \\
    \hline
    3 & ${\color{black}\mathrm{C}_j}+\mathrm{D}_k+\mathrm{E}_l$ & ${\color{black}\mathrm{A}_{\clubsuit}}+{\color{black}\mathrm{B}_i}+{\color{black}\mathrm{C}_j}+\mathrm{D}_k+\mathrm{E}_l$ & ${\color{black}\mathrm{A}_{\diamondsuit}}+{\color{black}\mathrm{B}_i}+{\color{black}\mathrm{C}_j}+\mathrm{D}_k+\mathrm{E}_l$ & ${\color{black}\mathrm{A}_{\spadesuit}}+{\color{black}\mathrm{B}_i}+{\color{black}\mathrm{C}_j}+\mathrm{D}_k+\mathrm{E}_l$ & $P_{1,2}\times P_1 \times \frac{1}{2} \times (\frac{1}{3})^4\times \frac{1}{6} = \frac{1}{5184}$ \\
    \hline
    4 & ${\color{black}\mathrm{C}_j}+\mathrm{D}_k+\mathrm{E}_l$ & ${\color{black}\mathrm{A}_{\clubsuit}}+\mathrm{D}_k+\mathrm{E}_l$ & ${\color{black}\mathrm{A}_{\diamondsuit}}+{\color{black}\mathrm{B}_i}+{\color{black}\mathrm{C}_j}+\mathrm{D}_k+\mathrm{E}_l$ & ${\color{black}\mathrm{A}_{\spadesuit}}+{\color{black}\mathrm{B}_i}+{\color{black}\mathrm{C}_j}+\mathrm{D}_k+\mathrm{E}_l$ & $P_{1,2}\times (1-P_1) \times \frac{1}{2} \times (\frac{1}{3})^4\times \frac{1}{6} = \frac{1}{5184}$ \\
    \hline
    5 & ${\color{black}\mathrm{B}_i}+{\color{black}\mathrm{C}_j}+\mathrm{D}_k$ & ${\color{black}\mathrm{A}_{\clubsuit}}+{\color{black}\mathrm{B}_i}+{\color{black}\mathrm{C}_j}+\mathrm{D}_k$ & ${\color{black}\mathrm{A}_{\diamondsuit}}+{\color{black}\mathrm{B}_i}+{\color{black}\mathrm{C}_j}+\mathrm{D}_k$ & ${\color{black}\mathrm{A}_{\spadesuit}}+{\color{black}\mathrm{B}_i}+{\color{black}\mathrm{C}_j}+\mathrm{D}_k$ & $P_{2,1}\times P_2 \times \frac{1}{2} \times (\frac{1}{3})^3\times \frac{1}{6} =0$ \\
    \hline
    6 & ${\color{black}\mathrm{B}_i}+{\color{black}\mathrm{C}_j}+\mathrm{D}_k$ & ${\color{black}\mathrm{A}_{\clubsuit}}+{\color{black}\mathrm{B}_i}+\mathrm{D}_k$ & ${\color{black}\mathrm{A}_{\diamondsuit}}+{\color{black}\mathrm{C}_j}+\mathrm{D}_k$ & ${\color{black}\mathrm{A}_{\spadesuit}}+{\color{black}\mathrm{B}_i}+{\color{black}\mathrm{C}_j}+\mathrm{D}_k$ & $P_{2,1}\times (1-P_2) \times \frac{1}{2} \times (\frac{1}{3})^3\times \frac{1}{6} = \frac{1}{864}$ \\
    \hline
    7 & ${\color{black}\mathrm{B}_i}+{\color{black}\mathrm{C}_j}+\mathrm{E}_l$ & ${\color{black}\mathrm{A}_{\clubsuit}}+{\color{black}\mathrm{B}_i}+{\color{black}\mathrm{C}_j}+\mathrm{E}_l$ & ${\color{black}\mathrm{A}_{\diamondsuit}}+{\color{black}\mathrm{B}_i}+{\color{black}\mathrm{C}_j}+\mathrm{E}_l$ & ${\color{black}\mathrm{A}_{\spadesuit}}+{\color{black}\mathrm{B}_i}+{\color{black}\mathrm{C}_j}+\mathrm{E}_l$ & $P_{2,1}\times P_2 \times \frac{1}{2} \times (\frac{1}{3})^3\times \frac{1}{6} = 0$ \\
    \hline
    8 & ${\color{black}\mathrm{B}_i}+{\color{black}\mathrm{C}_j}+\mathrm{E}_l$ & ${\color{black}\mathrm{A}_{\clubsuit}}+{\color{black}\mathrm{B}_i}+\mathrm{E}_l$ & ${\color{black}\mathrm{A}_{\diamondsuit}}+{\color{black}\mathrm{C}_j}+\mathrm{E}_l$ & ${\color{black}\mathrm{A}_{\spadesuit}}+{\color{black}\mathrm{B}_i}+{\color{black}\mathrm{C}_j}+\mathrm{E}_l$ & $P_{2,1}\times (1-P_2) \times \frac{1}{2} \times (\frac{1}{3})^3\times \frac{1}{6} = \frac{1}{864}$ \\
    \hline
    9 & ${\color{black}\mathrm{B}_i}+{\color{black}\mathrm{C}_j}+\mathrm{D}_k+\mathrm{E}_l$ & ${\color{black}\mathrm{A}_{\clubsuit}}+{\color{black}\mathrm{B}_i}+{\color{black}\mathrm{C}_j}+\mathrm{D}_k+\mathrm{E}_l$ & ${\color{black}\mathrm{A}_{\diamondsuit}}+{\color{black}\mathrm{B}_i}+{\color{black}\mathrm{C}_j}+\mathrm{D}_k+\mathrm{E}_l$ & ${\color{black}\mathrm{A}_{\spadesuit}}+{\color{black}\mathrm{B}_i}+{\color{black}\mathrm{C}_j}+\mathrm{D}_k+\mathrm{E}_l$ & $P_{2,2}\times P_2 \times (\frac{1}{3})^4\times \frac{1}{6} = 0$ \\
    \hline
    10 & ${\color{black}\mathrm{B}_i}+{\color{black}\mathrm{C}_j}+\mathrm{D}_k+\mathrm{E}_l$ & ${\color{black}\mathrm{A}_{\clubsuit}}+{\color{black}\mathrm{B}_i}+\mathrm{D}_k+\mathrm{E}_l$ & ${\color{black}\mathrm{A}_{\diamondsuit}}+{\color{black}\mathrm{C}_j}+\mathrm{D}_k+\mathrm{E}_l$ & ${\color{black}\mathrm{A}_{\spadesuit}}+{\color{black}\mathrm{B}_i}+{\color{black}\mathrm{C}_j}+\mathrm{D}_k+\mathrm{E}_l$ & $P_{2,2}\times (1-P_2) \times (\frac{1}{3})^4\times \frac{1}{6} = \frac{1}{2592}$ \\
    \hline
    \end{tabular}
    }
\end{table*}
}

Applying the proposed scheme in this example yields $11$ different types of answer sets $\{\mathrm{Z}_1,\mathrm{Z}_2,\mathrm{Z}_3,\mathrm{Z}_4\}$, where each $\mathrm{Z}_n$ for $n\in [4]$ is the answer of one of the servers to the user. 
Each type of answer set corresponds to specific values of $(I,J)$ and $\theta$, as defined in the proposed scheme.   
Table~\ref{tab:1} lists these $11$ types of answer sets and their corresponding probabilities, denoted by $P_{\{\mathrm{Z}_1,\mathrm{Z}_2,\mathrm{Z}_3,\mathrm{Z}_4\}}$, 
where  
$\mathrm{A}_{\clubsuit},\mathrm{A}_{\diamondsuit},\mathrm{A}_{\spadesuit}$ represent the three sub-packets of $\mathrm{A}$ in an arbitrary order, and 
$\mathrm{B}_i,\mathrm{C}_j,\mathrm{D}_k,\mathrm{E}_l$ each represent an arbitrary sub-packet of $\mathrm{B},\mathrm{C},\mathrm{D},\mathrm{E}$, respectively.
Note that each type contains multiple answer sets. 
In particular, there are $54$ answer sets of type~0 
(there are $6$ permutations of the sub-packets of $\mathrm{A}$, 
and 
$3$ choices each for $\mathrm{B}_i$ and $\mathrm{C}_j$);  
$486$ answer sets for each of types 1, 2, 3, and 4; 
$162$ answer sets for each of types 5, 6, 7, and 8; and 
$486$ answer sets for each of types~9 and 10.

The probability of each answer set depends on the values of $(I,J)$ and $\theta$ specifying the type of the answer set. 
In particular, given $(I,J)$ and $\theta$, 
the probability of an answer set $\{\mathrm{Z}_1,\mathrm{Z}_2,\mathrm{Z}_3,\mathrm{Z}_4\}$ of the corresponding type is given by
\[P_{\{\mathrm{Z}_1,\mathrm{Z}_2,\mathrm{Z}_3,\mathrm{Z}_4\}} = P_{I,J}\times P \times \frac{1}{\binom{2}{I}}\times \frac{1}{\binom{2}{J}} \times \left(\frac{1}{3}\right)^{2+J} \times \frac{1}{6},\] 
where $P=P_I$ (or $1-P_I$) when $\theta=0$ (or $1$). 
For this example, $P_{0,0} = 1/16$, $P_{1,2} = P_{2,1}=3/8$, and $P_{2,2} = 3/16$ by~\eqref{eq:Pij}, and $P_1=1/2$ and $P_2 = 0$ by~\eqref{eq:PI}. 
Note that ${{1}/{\binom{2}{I}}\times {1}/{\binom{2}{J}}}$ is the probability of selecting $I$ specific side information messages and $J$ specific interference messages; 
$(1/3)^{2+J}$ is the probability of choosing one specific sub-packet from each of the side information messages and 
the selected $J$ interference messages; 
and 
$1/6$ is the probability of selecting a specific permutation of the sub-packets of the demand message. 
Also, note that the total probability of all answer sets is equal to $54/864+4\times 486/5184+2\times 162/864+486/2592 = 1$. 

As evident from the table, the recoverability condition is met, as the demand message sub-packets, $\mathrm{A}_{\clubsuit},\mathrm{A}_{\diamondsuit},\mathrm{A}_{\spadesuit}$, can be recovered from any answer set of any type.  
To show that the privacy condition is met, 
we need to show that all answers with the same number of sub-packets have an equal probability. 

First, consider the answers $\mathrm{A}_m+\mathrm{B}_i+\mathrm{C}_j$ for $m\in [3]$. 
From the table, it is immediate that these answers have probability $6\times 1/864 = 1/144$. 
Now, consider the answer $\mathrm{B}_i+\mathrm{D}_k+\mathrm{E}_l$. 
This answer appears in two types of answer sets: types 1 and 2.
Each of these types contains $3\times 6=18$ answer sets, each including $\mathrm{B}_i+\mathrm{D}_k+\mathrm{E}_l$, corresponding to one of the choices of $C_j$ and one of the permutations of $\{\mathrm{A}_1,\mathrm{A}_2,\mathrm{A}_3\}$. 
Since each answer set of these types has probability $1/{5184}$, 
the probability of the answer $\mathrm{B}_i+\mathrm{D}_k+\mathrm{E}_l$ is $2\times 18\times 1/{5184} = 1/{144}$. 
Similarly, the probability of the answer $\mathrm{C}_j+\mathrm{D}_k+\mathrm{E}_l$ is $1/{144}$. 
Next, consider the answer $\mathrm{B}_i+\mathrm{C}_j+\mathrm{D}_k$. 
This answer appears in two types of answer sets: types 5 and 6. 
Each type contains $6$ answer sets, each including $\mathrm{B}_i+\mathrm{C}_j+\mathrm{D}_k$, corresponding to one of the permutations of $\{\mathrm{A}_1,\mathrm{A}_2,\mathrm{A}_3\}$. 
Since answer sets of type 5 have zero probability and answer sets of type 6 have probability $1/{864}$, 
the probability of the answer $\mathrm{B}_i+\mathrm{C}_j+\mathrm{D}_k$ is $6\times 1/{864} = 1/144$. 
Similarly, the probability of the answer $\mathrm{B}_i+\mathrm{C}_j+\mathrm{E}_l$ is $1/{144}$.

Now, consider the answers $\mathrm{A}_m+\mathrm{D}_k+\mathrm{E}_l$ for $m\in [3]$. 
Each of these answers appears in two types of answer sets: types 2 and 4. 
Each type contains $3\times 6 = 18$ answer sets, 
each including $\mathrm{A}_m+\mathrm{D}_k+\mathrm{E}_l$, 
corresponding to one of the choices of $B_i$ or $C_j$ and one of the permutations of $\{\mathrm{A}_1,\mathrm{A}_2,\mathrm{A}_3\}$. 
Since answer sets of these types have probability $1/{5184}$, 
the probability of each answer $\mathrm{A}_m+\mathrm{D}_k+\mathrm{E}_l$ for $m\in [3]$ is $2\times 18\times 1/{5184} = 1/{144}$.
Lastly, consider the answers $\mathrm{A}_m+\mathrm{B}_i+\mathrm{D}_k$ (or $\mathrm{A}_m+\mathrm{B}_i+\mathrm{E}_l$) for $m\in [3]$. 
Each of these answers appears in $3\times 2=6$ answer sets of type 6 (or type 8), 
each corresponding to one of the choices of $C_j$ and one of the permutations of $\{\mathrm{A}_1,\mathrm{A}_2,\mathrm{A}_3\}\setminus \{\mathrm{A}_m\}$. 
Thus, the probability of each answer $\mathrm{A}_m+\mathrm{B}_i+\mathrm{D}_k$ (or $\mathrm{A}_m+\mathrm{B}_i+\mathrm{E}_l$) for $m\in [3]$ is $6\times 1/{864} = 1/144$.
Similarly, the answers 
$\mathrm{A}_m+\mathrm{B}_i+\mathrm{E}_l$ (or $\mathrm{A}_m+\mathrm{C}_j+\mathrm{E}_l$) for $m\in [3]$ have probability $1/144$.

The above analysis verifies that all answers with $3$ sub-packets have an equal probability of $1/144$. 
Using a similar analysis, one can also verify that all answers with $4$ (or $5$) sub-packets have an equal probability of $1/432$ (or $13/1296$).  
These results confirm that the privacy condition is satisfied. 

The download cost of each answer set of type $0$ is $3/3 = 1$ (i.e., the user downloads a total of $3$ coded sub-packets from the servers, in order to retrieve the $3$ sub-packets of the demand message), whereas that of each answer set of any other type is $4/3$. 
Since there are $54$ answer sets of type $0$, each with probability $1/864$, 
the rate of the proposed scheme is equal to $1/(54/864+(1-54/864)\times 4/3) = 16/21$, which matches the rate $R$ defined in~\eqref{eq:R} for $N=4$, $K=5$, and $M=2$. 

\section{Proofs of Lemmas~\ref{lem:Pij} and~\ref{lem:Summij}}\label{sec:ProofsLem1Lem4}
In this section, we present the proofs of Lemmas~\ref{lem:Pij} and~\ref{lem:Summij}. 
Before stating the proofs, we present two key lemmas (Lemmas~\ref{lem:SumPolyDeg} and~\ref{lem:Sum}) 
that are useful in the proofs.  

\begin{lemma}\label{lem:SumPolyDeg}
For any polynomial $P(k)$ of degree less than $n$, 
\[\sum_{k=0}^{n} (-1)^k \binom{n}{k} P(k) = 0.\]
\end{lemma}

\begin{proof}
Any polynomial $P(k)$ of degree less than $n$ can be written as a linear combination of $1,k,\dots,k^{n-1}$. 
Thus, it suffices to show that $\sum_{k=0}^n (-1)^k \binom{n}{k} k^i = 0$ for all ${i\in [0:n-1]}$.

By the binomial identity, we have $\sum_{k=0}^{n} \binom{n}{k} x^k = (1+x)^n$. 
Setting $x=-1$, it follows that $\sum_{k=0}^{n} (-1)^k \binom{n}{k} = 0$, which proves the result for $i=0$. 
Fix an arbitrary $i\in [n-1]$. 
Assume $\sum_{k=0}^n (-1)^k \binom{n}{k} k^j = 0$ for all $j\in [0:i-1]$. 
We need to show that 
$\sum_{k=0}^n (-1)^k \binom{n}{k} k^i = 0$. 

Taking the derivative of both sides of the binomial identity $i$ times with respect to $x$, 
\[\sum_{k=0}^{n} \binom{n}{k} P_i(k) x^{k-i} = {n(n-1)\cdots(n-i+1) (1+x)^{n-i}},\] 
where $P_i(k) := k(k-1)\cdots (k-i+1)$.  
Setting $x=-1$, 
\begin{equation}\label{eq:A1}
\sum_{k=0}^{n}  (-1)^{k} \binom{n}{k} P_i(k) = 0.
\end{equation}

Since $P_i(k)$ is a polynomial in $k$ of degree $i$, it can be expressed as a linear combination of $1,k,\dots,k^{i-1},k^{i}$, say, $P_i(k) = \sum_{j=0}^{i-1} a_j k^j + k^i$. 
Substituting $P_i(k)$ into~\eqref{eq:A1}, we have 
\begin{equation}\label{eq:A2}
\sum_{j=0}^{i-1} \sum_{k=0}^{n} (-1)^k \binom{n}{k} a_j k^j + \sum_{k=0}^{n} (-1)^k \binom{n}{k} k^i= 0.
\end{equation}
Since the double sum equals zero by assumption,  
the second sum also equals zero, 
as was to be shown. 
\end{proof}

\begin{lemma}\label{lem:Sum}
For any $1\leq j\leq K-M-1$, \[\sum_{k=(j-M)^{+}}^{j} (-1)^k \binom{M}{j-k}\binom{M+k-1}{k} = 0.\] 
\end{lemma}

\begin{proof}
For any $1\leq j\leq M$, we need to show that \[\sum_{k=0}^{j} (-1)^k \binom{M}{j-k}\binom{M+k-1}{k} = 0.\] 
It is easy to verify that \[\binom{M}{j-k}\binom{M+k-1}{k} = \frac{M}{j}\binom{j}{k}\binom{M+k-1}{j-1}.\]  
Thus, 
\begin{align}\label{eq:Lem5eq1}
& \sum_{k=0}^{j} (-1)^k \binom{M}{j-k}\binom{M+k-1}{k} \nonumber \\
& = \frac{M}{j}\sum_{k=0}^{j} (-1)^k \binom{j}{k}\binom{M+k-1}{j-1}.
\end{align}
Since $\binom{M+k-1}{j-1}$ is a polynomial in $k$ of degree $j-1$, 
the RHS of~\eqref{eq:Lem5eq1} is equal to zero by the result of Lemma~\ref{lem:SumPolyDeg}.

For any $M+1\leq j\leq K-M-1$, we need to show that \[\sum_{k=j-M}^{j} (-1)^k \binom{M}{j-k}\binom{M+k-1}{k} = 0.\] 
Substituting $k$ with $k-j+M$, it is easy to verify that it suffices to show that $\sum_{k=0}^{M} (-1)^k \binom{M}{k}\binom{k+j-1}{M-1} = 0$. 
Note that $\binom{k+j-1}{M-1}$ is a polynomial in $k$ of degree $M-1$. 
Thus, $\sum_{k=0}^{M} (-1)^k \binom{M}{k}\binom{k+j-1}{M-1} = 0$ by Lemma~\ref{lem:SumPolyDeg}.
\end{proof}

\subsection{Proof of Lemma~\ref{lem:Pij}}
First, we show that for $P_{i,j}$'s defined in~\eqref{eq:Pij}, 
$P_{i,j}\geq 0$ for all $0\leq i\leq M$ and $0\leq j\leq K-M-1$. 
This is equivalent to show that $m_{i,j}\geq 0$ for all $i$ and $j$. 
By definition, $m_{0,0}=1$, and $m_{i,j}=0$ for all $i$ and $j$ such that $1\leq i+j\leq M$. 
Thus, it remains to show that $m_{i,j}\geq 0$ for all $i$ and $j$ such that $M+1\leq i+j\leq K-1$. 
According to~\eqref{eq:mij},
we need to show that ${\sum_{k=0}^{s} (-1)^k \binom{\alpha+k}{k} \beta^{-k}\geq 0}$ for all $s\geq 0$, where $\alpha:=M-1$, $\beta:=N-1$, and $s:=i+j-M-1$.    

Note that 
$\sum_{k=0}^{s} (-1)^k \binom{\alpha+k}{k} \beta^{-k}$ can be expressed as 
\[\sum_{l=0}^{\frac{s}{2}-1} \beta^{-2l} \left[\binom{\alpha+2l}{2l}-\binom{\alpha+2l+1}{2l+1} \beta^{-1}\right] + \binom{\alpha+s}{s}\beta^{-s}\] for even $s$, and 
\[\sum_{l=0}^{\frac{s-1}{2}} \beta^{-2l} \left[\binom{\alpha+2l}{2l}-\binom{\alpha+2l+1}{2l+1} \beta^{-1}\right]\] for odd $s$. 
Thus, it suffices to show that \[\binom{\alpha+2l}{2l}-\binom{\alpha+2l+1}{2l+1} \beta^{-1}\geq 0\] for all ${l\geq 0}$. 
This is evident because
\begin{align*}
& \binom{\alpha+2l}{2l}-\binom{\alpha+2l+1}{2l+1} \beta^{-1} \\ 
& \quad = \binom{\alpha+2l}{2l} \left(\beta-\frac{\alpha}{l+1}-1\right)\beta^{-1},
\end{align*}
and 
$\beta-\frac{\alpha}{l+1}-1\geq 0$ for all $l\geq 0$, 
noting that 
${\alpha\geq 0}$, ${\beta\geq 1}$, and 
${\beta\geq \alpha+1}$,  
since ${M\geq 1}$, ${N>1}$, and ${N\geq M+1}$. 
This completes the proof of the first part of Lemma~\ref{lem:Pij}. 

Next, we prove the second part of Lemma~\ref{lem:Pij}. 
We need to show that 
$\sum_{i=0}^{M}\sum_{j=0}^{K-M-1} P_{i,j} = 1$. 
According to~\eqref{eq:Pij}, this is equivalent to show that $\sum_{i=0}^{M}\sum_{j=0}^{K-M-1} m_{i,j} = N^{K-M-1}$. 

First, consider the case of ${K-M-1\leq M}$. 
In this case, we need to show that 
\begin{equation}\label{eq:SSS1}
1+\sum_{j=1}^{K-M-1}\binom{K-M-1}{j} \sum_{i=M-j+1}^{M}\binom{M}{i} m_{i,j}  = N^{K-M-1},
\end{equation}
noting that $m_{0,0} = 1$ and $m_{i,j} = 0$ for all $i,j$ such that ${1\leq i+j\leq M}$. 
To prove~\eqref{eq:SSS1}, 
it suffices to show that 
\begin{equation}\label{eq:SS1}
\sum_{i=M-j+1}^{M}\binom{M}{i} m_{i,j} = (N-1)^j    
\end{equation}
for any $1\leq j\leq K-M-1$. 
This is because by using~\eqref{eq:SS1}, 
we can rewrite~\eqref{eq:SSS1} as \[1+\sum_{j=1}^{K-M-1}\binom{K-M-1}{j} (N-1)^j = N^{K-M-1},\] which follows directly from the binomial identity. 

For all $1\leq j\leq K-M-1$ and $M-j+1\leq i\leq M$, 
we have $M+1\leq i+j\leq M+j\leq K-1$. 
Thus, by the definition of $m_{i,j}$, 
$\sum_{i=M-j+1}^{M}\binom{M}{i} m_{i,j}$ is given by
\begin{align}\label{eq:SS1E1}
& \sum_{i=M-j+1}^{M} \binom{M}{i} \sum_{k=0}^{i+j-M-1} (-1)^k \nonumber \\
& \quad \times \binom{M+k-1}{k} (N-1)^{i+j-M-k}.
\end{align} 
Substituting $i$ with $M-i$, we can rewrite~\eqref{eq:SS1E1} as
\begin{align}\label{eq:SS1E2}
\sum_{i=0}^{j-1}\binom{M}{i} \sum_{k=0}^{j-i-1} (-1)^k \binom{M+k-1}{k} (N-1)^{j-i-k}.
\end{align}
Letting 
$l_k:=(-1)^k \binom{M+k-1}{k} (N-1)^{j-i-k}$, we can readily simplify~\eqref{eq:SS1E2} as
\[\sum_{i=0}^{j-1} \sum_{k=0}^{j-i-1} \binom{M}{i} l_k,\] 
or equivalently, \begin{equation}\label{eq:SS1I1}
\sum_{i=0}^{j-1} \sum_{k=0}^{j-i} \binom{M}{i} l_k - \sum_{i=0}^{j-1} \binom{M}{i} l_{j-i}
\end{equation} 
by adding and subtracting the term $l_{j-i} = (-1)^{j-i}\binom{M+j-i-1}{j-i}$. 
Thus, we have
\begin{align}\label{eq:SS1E3}
\sum_{i=M-j+1}^{M}\binom{M}{i} m_{i,j} = \sum_{i=0}^{j-1} \sum_{k=0}^{j-i} \binom{M}{i} l_k - \sum_{i=0}^{j-1} \binom{M}{i} l_{j-i}.
\end{align}
Note that
\[\sum_{i=0}^{j-1} \binom{M}{i} l_{j-i} = \sum_{i=0}^{j} \binom{M}{i} l_{j-i} - \binom{M}{j}l_0.\] 
Moreover, 
\[\sum_{i=0}^{j} \binom{M}{i} l_{j-i} = \sum_{i=0}^{j} (-1)^{j-i} \binom{M}{i}\binom{M+j-i-1}{j-i},\] 
which can be rewritten as 
\begin{align*}
\sum_{i=0}^{j} \binom{M}{i} l_{j-i} = \sum_{k=0}^{j} (-1)^{k} \binom{M}{j-k}\binom{M+k-1}{k}
\end{align*} 
by substituting $j-i$ with $k$. 
Since the RHS is equal to zero by Lemma~\ref{lem:Sum}, 
the LHS is also equal to zero,
implying that 
\[\sum_{i=0}^{j-1} \binom{M}{i} l_{j-i} = -\binom{M}{j}l_0.\] 
Thus, we can rewrite~\eqref{eq:SS1E3} as
\begin{equation}\label{eq:SS1I2}
\sum_{i=M-j+1}^{M}\binom{M}{i} m_{i,j} = \sum_{i=0}^{j-1} \sum_{k=0}^{j-i} \binom{M}{i} l_k + \binom{M}{j}l_0.
\end{equation}
Since 
\[\sum_{i=0}^{j-1} \sum_{k=0}^{j-i} \binom{M}{i} l_k = \sum_{i=0}^{j} \sum_{k=0}^{j-i} \binom{M}{i} l_k - \binom{M}{j}l_0,\] 
we can simplify~\eqref{eq:SS1I2} to 
\begin{equation}\label{eq:SS1I3}
\sum_{i=M-j+1}^{M}\binom{M}{i} m_{i,j} = \sum_{i=0}^{j} \sum_{k=0}^{j-i} \binom{M}{i} l_k.
\end{equation}

Comparing~\eqref{eq:SS1} and~\eqref{eq:SS1I3}, it is evident that we need to show $\sum_{i=0}^{j} \sum_{k=0}^{j-i} \binom{M}{i} l_k = (N-1)^j$, or equivalently,  
\begin{equation}\label{eq:SS1I4}
\sum_{i=0}^{j} \sum_{k=0}^{j-i} (-1)^k \binom{M}{i}\binom{M+k-1}{k} (N-1)^{-i-k} = 1
\end{equation} 
for any $1\leq j\leq K-M-1$. 

Substituting $k$ with $k-i$, 
the LHS of~\eqref{eq:SS1I4} can be expressed as  
\[\sum_{i=0}^{j} \sum_{k=i}^{j} (-1)^{k-i} \binom{M}{i}\binom{M+k-i-1}{k-i} (N-1)^{-k},\] 
or equivalently, 
\[\sum_{k=0}^{j} \sum_{i=0}^{k} (-1)^{k-i} \binom{M}{i}\binom{M+k-i-1}{k-i} (N-1)^{-k}\]
by reordering the sums. 
Substituting $i$ with $k-i$ and interchanging the variables $i$ and $k$, 
we can rewrite this as 
\[\sum_{i=0}^{j} \sum_{k=0}^{i} (-1)^{k} \binom{M}{i-k}\binom{M+k-1}{k} (N-1)^{-i}.\] 
It is easy to see that 
\begin{align*}
& \sum_{i=0}^{j} \sum_{k=0}^{i} (-1)^{k} \binom{M}{i-k}\binom{M+k-1}{k} (N-1)^{-i} \\
& \quad = 1 + \sum_{i=1}^{j} \sum_{k=0}^{i} (-1)^{k} \binom{M}{i-k}\binom{M+k-1}{k} (N-1)^{-i}.
\end{align*} 
Thus, in order to prove~\eqref{eq:SS1I4}, 
it remains to show that 
\begin{align}\label{eq:SS1E4}
\sum_{i=1}^{j} \sum_{k=0}^{i} (-1)^{k} \binom{M}{i-k}\binom{M+k-1}{k} (N-1)^{-i} = 0.
\end{align}
For any $1\leq i\leq j$, 
we have $1\leq i\leq K-M-1$ and $(i-M)^{+} = 0$ 
since $j\leq K-M-1$ and $K-M-1\leq M$. 
Thus, applying  Lemma~\ref{lem:Sum} yields
\[\sum_{k=0}^{i} (-1)^{k} \binom{M}{i-k}\binom{M+k-1}{k} = 0,\] 
which readily implies~\eqref{eq:SS1E4}, 
thereby completing the proof of~\eqref{eq:SS1}, and subsequently,~\eqref{eq:SSS1}. 

Next, consider the case of $K-M-1>M$. 
In this case, we need to show that 
\begin{align}\label{eq:SSS2}
& 1+\sum_{j=1}^{M}\binom{K-M-1}{j} \sum_{i=M-j+1}^{M}\binom{M}{i} m_{i,j} \nonumber \\
& + \sum_{j=M+1}^{K-M-1}\binom{K-M-1}{j} \sum_{i=0}^{M}\binom{M}{i} m_{i,j} = N^{K-M-1}.
\end{align}
According to~\eqref{eq:SS1}, 
we have $\sum_{i=M-j+1}^{M} \binom{M}{i} m_{i,j} = (N-1)^j$ for any $1\leq j\leq M$ ($<K-M-1$). 
Thus, in order to prove~\eqref{eq:SSS2}, 
it remains to show that 
\begin{equation}\label{eq:SS2}
\sum_{i=0}^{M}\binom{M}{i} m_{i,j} = (N-1)^j 
\end{equation} 
for any $M+1\leq j\leq K-M-1$. 
This is because by applying~\eqref{eq:SS1} and~\eqref{eq:SS2}, the LHS of~\eqref{eq:SSS2} can be rewritten as 
\begin{dmath*}
1+\sum_{j=1}^{M} \binom{K-M-1}{j} (N-1)^j +\sum_{j=M+1}^{K-M-1} \binom{K-M-1}{j} (N-1)^j,
\end{dmath*}
or equivalently, 
\[1+\sum_{j=1}^{K-M-1} \binom{K-M-1}{j} (N-1)^j,\] 
which is equal to the RHS of~\eqref{eq:SSS2}, i.e., $N^{K-M-1}$, 
by the binomial identity.  

By the definition of $m_{i,j}$, 
$\sum_{i=0}^{M}\binom{M}{i} m_{i,j}$ is given by 
\begin{align}\label{eq:SS2E1}
& \sum_{i=0}^{M}\binom{M}{i} \sum_{k=0}^{i+j-M-1} (-1)^k \binom{M+k-1}{k} (N-1)^{i+j-M-k}
\end{align} 
Letting $l_k:=(-1)^k \binom{M+k-1}{k} (N-1)^{i+j-M-k}$, 
we can rewrite~\eqref{eq:SS2E1} as
\[\sum_{i=0}^{M}\binom{M}{i} \sum_{k=0}^{i+j-M-1} l_k,\] 
or equivalently, 
\begin{align}\label{eq:SS2E2}
\sum_{i=0}^{M}\binom{M}{i} \sum_{k=0}^{i+j-M} l_k - \sum_{i=0}^{M} \binom{M}{i} l_{i+j-M}.
\end{align}
Thus, we have 
\begin{align}\label{eq:SS2E3}
\sum_{i=0}^{M}\binom{M}{i} m_{i,j} = \sum_{i=0}^{M}\binom{M}{i} \sum_{k=0}^{i+j-M} l_k - \sum_{i=0}^{M} \binom{M}{i} l_{i+j-M}.
\end{align}

Note that 
\[\sum_{i=0}^{M} \binom{M}{i} l_{i+j-M} = \sum_{i=0}^{M} (-1)^{i+j-M} \binom{M}{i} \binom{i+j-1}{M-1},\]
which can be rewritten as 
\[\sum_{i=0}^{M} \binom{M}{i} l_{i+j-M} = \sum_{k=j-M}^{j} (-1)^{k} \binom{M}{j-k} \binom{M+k-1}{k}\]
by substituting $i$ with $j-k$. 
Since $(j-M)^{+} = j-M$ for any $M+1\leq j\leq K-M-1$, 
the RHS is equal to zero by Lemma~\ref{lem:Sum}, 
implying that the LHS is also equal to zero,
\[\sum_{i=0}^{M} \binom{M}{i} l_{i+j-M}=0.\]
Thus, we can simplify~\eqref{eq:SS2E3} to 
\begin{equation}\label{eq:SS2I1}
\sum_{i=0}^{M}\binom{M}{i} m_{i,j} 
= \sum_{i=0}^{M}\binom{M}{i} \sum_{k=0}^{i+j-M} l_k.
\end{equation}

Comparing~\eqref{eq:SS2} and~\eqref{eq:SS2I1}, it is evident that we need to show that $\sum_{i=0}^{M}\binom{M}{i} \sum_{k=0}^{i+j-M} l_k = (N-1)^j$, or equivalently,
\begin{equation}\label{eq:SS2I2}
\sum_{i=0}^{M} \sum_{k=0}^{i+j-M} (-1)^k \binom{M}{i}\binom{M+k-1}{k} (N-1)^{i-M-k} = 1
\end{equation} 
for any $M+1\leq j\leq K-M-1$. 

Substituting $k$ with $i+j-M-k$, 
the LHS of~\eqref{eq:SS2I2} becomes
\[\sum_{i=0}^{M} \sum_{k=0}^{i+j-M} (-1)^{i+j-M-k} \binom{M}{i}\binom{i+j-k-1}{M-1} (N-1)^{k-j}.\] 
By reordering the sums, 
we can rewrite this expression as the sum of two double sums: 
\begin{align*}
& \sum_{k=0}^{j-M-1} \sum_{i=0}^{M} (-1)^{i+j-M-k} \binom{M}{i} \\ 
& \quad \times \binom{i+j-k-1}{M-1} (N-1)^{k-j},
\end{align*}
and 
\begin{dmath*}
\sum_{k=j-M}^{j} \sum_{i=M-j+k}^{M} (-1)^{i+j-M-k} \binom{M}{i} \times \binom{i+j-k-1}{M-1} (N-1)^{k-j}.
\end{dmath*}

We can express the first double sum as 
\[\sum_{i=0}^{j-M-1}\hspace{-0.15cm} (-1)^{j-M-i} (N-1)^{i-j} \sum_{k=0}^{M} (-1)^{k} \binom{M}{k}\binom{k+j-i-1}{M-1}\] 
by interchanging the variables $i$ and $k$. 
Since $\binom{k+j-i-1}{M-1}$ is a polynomial in $k$ of degree $M-1$, the inner sum is equal to zero by Lemma~\ref{lem:SumPolyDeg}, 
which implies that the first double sum is equal to zero. 

Substituting $i$ with $i+M-j+k$ and 
interchanging the variables $i$ and $k$, 
the second double sum can be rewritten as 
\[\sum_{i=j-M}^{j} (N-1)^{i-j} \sum_{k=0}^{j-i} (-1)^{k} \binom{M}{j-i-k}\binom{M+k-1}{k},\] 
or equivalently, 
\[1+\sum_{i=j-M}^{j-1} (N-1)^{i-j} \sum_{k=0}^{j-i} (-1)^{k} \binom{M}{j-i-k}\binom{M+k-1}{k}.\]
For any $j-M\leq i\leq j-1$, 
we have ${1\leq j-i<K-M-1}$ and 
${(j-i-M)^{+} = 0}$ since $j-i\leq M$ and $M<K-M-1$. 
Thus, 
using Lemma~\eqref{lem:Sum}, 
the inner sum is equal to zero, 
implying that the second double sum is equal to $1$. 
This directly implies~\eqref{eq:SS2I2}, 
which completes the proof of~\eqref{eq:SS2}, and consequently,~\eqref{eq:SSS2}. 

\subsection{Proof of Lemma~\ref{lem:Summij}}
Fix an arbitrary $1\leq j\leq K-M-1$. 
We need to show that 
\begin{align}\label{eq:L0}
\sum_{i = (M-j+1)^{+}}^{M} \left(\binom{M}{i} (N-1) - \binom{M}{i-1} \right) m_{i,j} = m_{M,j+1}
\end{align} 
where 
\begin{align}\label{eq:mijnew}
m_{i,j} = \sum_{k=0}^{i+j-M-1}(-1)^k \binom{M+k-1}{k}(N-1)^{i+j-M-k},
\end{align} 
since $M+1\leq i+j\leq K-1$ for all $(M-j+1)^{+}\leq i\leq M$. 
We analyze the cases of $j\leq M$ and $j>M$ separately. 

First, consider the case of $j\leq M$. 
The LHS of~\eqref{eq:L0} can be rewritten as 
\begin{align}\label{eq:L1}
& \hspace{-0.25cm} \sum_{i = M-j+1}^{M} \sum_{k=0}^{i+j-M-1} (-1)^k \binom{M+k-1}{k} \nonumber \\ 
& \quad \times (N-1)^{i+j-M-k}  \left(\binom{M}{i} (N-1) - \binom{M}{i-1} \right),   
\end{align} 
by using~\eqref{eq:mijnew} and noting that $(M-j+1)^{+} = M-j+1$ for $j\leq M$. 
Substituting $i$ with $i+M-j$ and reordering the sums, 
we can rewrite~\eqref{eq:L1} as 
\begin{align}\label{eq:L2}
& \sum_{k = 0}^{j-1} (-1)^k \binom{M+k-1}{k} (N-1)^{-k} \nonumber \\
& \quad \times \sum_{i=k+1}^{j} \left(\binom{M}{j-i} (N-1)^{i+1} - \binom{M}{j-i+1}(N-1)^i \right).
\end{align}
Since the inner sum is a telescoping sum,   
simplifying to ${(N-1)^{j+1} - \binom{M}{j-k}(N-1)^{k+1}}$, 
we can express~\eqref{eq:L2} as 
\begin{align*}
& \sum_{k = 0}^{j} (-1)^k \binom{M+k-1}{k} (N-1)^{j-k+1} \nonumber \\
& \quad - \sum_{k = 0}^{j} (-1)^k \binom{M}{j-k} \binom{M+k-1}{k} (N-1),
\end{align*} 
by adding and subtracting the term $(-1)^j\binom{M+j-1}{j} (N-1)$. 
According to~\eqref{eq:mijnew}, 
the first sum is equal to $m_{M,j+1}$, i.e., the RHS of~\eqref{eq:L0}, and  
the second sum is equal to zero
by the result of Lemma~\ref{lem:Sum}, 
noting that $(j-M)^{+} = 0$ 
for $j\leq M$.  
This completes the proof for the case of $j\leq M$. 

Next, consider the case of $j>M$. 
By applying~\eqref{eq:mijnew} and noting that $(M-j+1)^{+} = 0$ for $j>M$, 
we can rewrite the LHS of~\eqref{eq:L0} as
\begin{align}\label{eq:L4}
& \sum_{i = 0}^{M} \sum_{k=0}^{i+j-M-1} (-1)^k \binom{M+k-1}{k} \nonumber \\ 
& \quad \times (N-1)^{i+j-M-k}  \left(\binom{M}{i} (N-1) - \binom{M}{i-1} \right).   
\end{align} 
Substituting $i$ with $i+M-j$, 
reordering the sums, 
and rearranging terms, 
we can express~\eqref{eq:L4} as the sum of the following two double sums: 
\begin{align*}
& \sum_{k = 0}^{j-M-1} (-1)^k \binom{M+k-1}{k} (N-1)^{-k} \nonumber \\
& \quad \times \sum_{i=j-M}^{j} \left(\binom{M}{j-i} (N-1)^{i+1} - \binom{M}{j-i+1}(N-1)^i \right), 
\end{align*} and 
\begin{align*}
& \sum_{k = j-M}^{j-1} (-1)^k \binom{M+k-1}{k} (N-1)^{-k} \nonumber \\
& \quad \times \sum_{i=k+1}^{j} \left(\binom{M}{j-i} (N-1)^{i+1} - \binom{M}{j-i+1}(N-1)^i \right). 
\end{align*}
Since both inner sums are telescoping sums, with the first one simplifying to $(N-1)^{j+1} $ and the second one to ${(N-1)^{j+1}-\binom{M}{j-k}(N-1)^{k+1}}$, 
we can rewrite the first double sum as 
\begin{align}\label{eq:L5}
& \sum_{k = 0}^{j-M-1} (-1)^k \binom{M+k-1}{k} (N-1)^{j-k+1}, 
\end{align} and the second one as 
\begin{align}\label{eq:L6}
& \sum_{k = j-M}^{j} (-1)^k \binom{M+k-1}{k} (N-1)^{j-k+1} \nonumber \\
& \quad - \sum_{k = j-M}^{j} (-1)^k \binom{M}{j-k} \binom{M+k-1}{k} (N-1), 
\end{align} 
by adding and subtracting the term $(-1)^j\binom{M+j-1}{j} (N-1)$. 
Summing~\eqref{eq:L5} and~\eqref{eq:L6}, 
we can express~\eqref{eq:L4} as 
\begin{align*}
& \sum_{k = 0}^{j} (-1)^k \binom{M+k-1}{k} (N-1)^{j-k+1} \nonumber \\
& \quad - \sum_{k = j-M}^{j} (-1)^k \binom{M}{j-k} \binom{M+k-1}{k} (N-1). 
\end{align*} 
According to~\eqref{eq:mijnew}, 
the first sum is equal to 
$m_{M,j+1}$, i.e., the RHS of~\eqref{eq:L0}, and 
the second sum is equal to zero due to Lemma~\ref{lem:Sum}, 
noting that $(j-M)^{+} = j-M$ 
for $j>M$.  
This completes the proof for the case of $j>M$. 

\end{appendices}

\end{document}